\documentclass[11pt]{article}
\usepackage[margin=1in]{geometry}
\usepackage{amssymb} 
\usepackage{amsmath} 
\usepackage{amsthm}
\usepackage[utf8]{inputenc} 
\usepackage{multirow} 
\usepackage{xcolor}
\usepackage{xspace}
\usepackage{thm-restate} 
\usepackage{hyperref}
\hypersetup{
    colorlinks=true,
    allcolors=blue,
}
\usepackage{algorithm}
\usepackage[noend]{algpseudocode}
\usepackage{mathtools}
\usepackage{footnote}
\makesavenoteenv{tabular}
\usepackage{graphicx}
\usepackage{subcaption}

\algblockdefx[Parallel]{Parallel}{EndParallel}[1]{\textbf{in parallel} #1 \textbf{times do}}{}
\algblockdefx[Pass]{Pass}{EndPass}[1]{\textbf{in pass} #1 \textbf{do}}{}
\algblockdefx[Stage]{Stage}{EndStage}[1]{\textbf{in stage} #1 \textbf{do}}{}

\makeatletter
\ifthenelse{\equal{\ALG@noend}{t}}%
  {\algtext*{EndPass}}
  {}%
\ifthenelse{\equal{\ALG@noend}{t}}%
  {\algtext*{EndParallel}}
  {}%
\ifthenelse{\equal{\ALG@noend}{t}}%
  {\algtext*{EndStage}}
  {}%
\makeatother

\numberwithin{equation}{section} 

\newtheorem{theorem}{Theorem}[section]

\newtheorem{lemma}[theorem]{Lemma}
\newtheorem{fact}[theorem]{Fact}
\newtheorem{proposition}[theorem]{Proposition}
\theoremstyle{definition}
\newtheorem{definition}[theorem]{Definition}
\newtheorem*{remark}{Remark}

\newcommand{\ip}[2]{\langle a_{#1}, a_{#2}\rangle}
\newcommand{\tildeip}[2]{\langle \tilde{a}_{#1}, \tilde{a}_{#2} \rangle}
\newcommand{\Exp}[2]{\mathbb{E}_{#1}\left[ #2 \right]}
\newcommand{\Prob}[1]{\mathbb{P}\left[ #1 \right]}
\newcommand{\R}[0]{\mathbb{R}}
\newcommand{\spn}[2]{\| #1 \|_{S_{#2}}^{#2}}
\newcommand{\ceil}[1]{\lceil #1 \rceil}
\newcommand{\floor}[1]{\lfloor #1 \rfloor}
\newcommand{\pluseq}{\mathrel{+}=}
\DeclareRobustCommand{\stirling}{\genfrac\{\}{0pt}{}}
\newcommand{\Z}[0]{\mathbb{Z}}
\newcommand{\N}[0]{\mathbb{N}}
\newcommand{\tr}[0]{\mathsf{Tr}}

\newcommand{\Sc}{\mathcal{S}}
\newcommand{\Lc}{\mathcal{L}}
\newcommand{\Cc}{\mathcal{C}}

\newcommand{\supp}[1]{\operatorname{supp}(#1)}
\newcommand{\set}[1]{\{#1\}}
\newcommand{\abs}[1]{\lvert#1\rvert}

\DeclareMathOperator{\Var}{Var}
\DeclareMathOperator{\poly}{poly}

\newcommand\tuple[1]{\langle {#1}\rangle}
\providecommand{\eqdef}{\coloneqq}
\providecommand{\set}[1]{{\{#1\}}}
\renewcommand{\epsilon}{\varepsilon} 

\usepackage[colorinlistoftodos,textsize=tiny,textwidth=2cm,color=red!25!white,obeyFinal]{todonotes}

\begin{document}
\title{Schatten Norms in Matrix Streams: Hello Sparsity, Goodbye
              Dimension}
\author{Vladimir Braverman%
  \thanks{Johns Hopkins University.
This research was supported in part by NSF CAREER grant 1652257, ONR Award N00014-18-1-2364 and and the Lifelong Learning Machines program from DARPA/MTO.
Email: \texttt{vova@cs.jhu.edu}
}
\and
Robert Krauthgamer%
\thanks{Weizmann Institute of Science.
  Work partially supported by ONR Award N00014-18-1-2364, the Israel Science Foundation grant \#1086/18, and a Minerva Foundation grant.
  Part of this work was done while the author was visiting the Simons Institute for the Theory of Computing.
  Email: \texttt{\{robert.krauthgamer,roi.sinoff\}@weizmann.ac.il}
}
\and
Aditya Krishnan%
\thanks{Johns Hopkins University. Work partially supported by ONR Award N00014-18-1-2364. Email: \texttt{akrish23@jhu.edu}
}
\and
Roi Sinoff\footnotemark[2]
}


\maketitle

\begin{abstract}
Spectral functions of large matrices contains important structural information about the underlying data, and is thus becoming increasingly important. 
Many times, large matrices representing real-world data are
\emph{sparse} or \emph{doubly sparse} (i.e., sparse in both rows and columns),
and are accessed as a \emph{stream} of updates,
typically organized in \emph{row-order}.
In this setting, where space (memory) is the limiting resource,
all known algorithms require space that is polynomial
in the dimension of the matrix, even for sparse matrices. 
We address this challenge by providing the first algorithms
whose space requirement is \emph{independent of the matrix dimension},
assuming the matrix is doubly-sparse and presented in row-order. 
Our algorithms approximate the Schatten $p$-norms, 
which we use in turn to approximate other spectral functions,
such as logarithm of the determinant, trace of matrix inverse,
and Estrada index.
We validate these theoretical performance bounds
by numerical experiments on real-world matrices representing social networks. 
We further prove that multiple passes are unavoidable in this setting,
and show extensions of our primary technique,
including a trade-off between space requirements and number of passes. 
\end{abstract}

\section{Introduction}

Large matrices are often used to represent real-world data sets like text documents, images and social networks,
however analyzing them is increasingly challenging,
as their sheer size renders many algorithms impractical.
Fortunately,
in several application domains, input matrices are often very sparse,
meaning that only a small fraction of their entries are non-zero.
In fact, in applications related to natural language processing (e.g.~\cite{GDC13}),
image recognition, medical imaging and computer vision (e.g.~\cite{GJPKI12,LZYQ15}),
the matrices are often doubly sparse, i.e., sparse in both rows and columns.
Throughout, we define these matrices as \emph{$k$-sparse}, meaning
that every row and every column has at most $k$ non-zero entries.
The current work devises new algorithms to analyze
the \emph{spectrum} (singular values) of such sparse matrices, 
aiming to achieve efficiency (storage requirement in streaming model)
that depends on \emph{matrix sparsity}  instead of \emph{matrix dimension}. 

We focus on fundamental functions of the spectrum, called the Schatten norms. 
Formally, the Schatten $p$-norm of a matrix $A \in \R^{m \times n}, m\geq n$ with singular values $\sigma_1 \geq \ldots \geq \sigma_n \geq 0$ is defined for every $p\geq 1$ as
$$
  \|A\|_{S_p}
  \eqdef
  \left(\sum_{i = 1}^n \sigma_i^p \right)^{1/p}.$$
This definition extends also to $0<p<1$, in which case it is not a norm, and also to $p=0,\infty$ by taking the limit. Frequently used cases include $p=0$, representing the rank of $A$, and $p=1,2,\infty$, commonly known as the trace/nuclear norm, Frobenius norm, and spectral/operator norm, respectively. Schatten norms are often used as surrogates for the spectrum, as explained in~\cite{ZWJ15,WV16,DNPS16,KO19}, and some specific cases have applications in optimization, image processing, and differential privacy etc.~\cite{XGLZZZ16,MNSUW18}.

For a positive semidefinite (PSD) matrix $A$, 
the Schatten norms can be easily used to approximate
other important spectral functions. 
One example is the Estrada index, which has applications in chemistry, physics, network theory and information theory (see survey by Gutman et al.~\cite{GDR11}).
Another example is the trace of matrix inverse,
which is used for image restoration, counting triangles in graphs,
to measure the uncertainty in data collections, and to bound the total variance of unbiased estimators
(see e.g. \cite{WLKSG16, Che16, HMAS17} for references).
A third example is the logarithm of the determinant,
used in many machine learning tasks,
like Bayesian learning, kernel learning, Gaussian processes, tree mixture models,
spatial statistics and Markov field models
(see e.g. \cite{HMS15, UCS17, US17, HMAS17} for references).
Thus, our results for Schatten norm have further applications. 

As the matrices in many real-world applications are often very large, storing the entire matrices in working memory can be impractical, and thus, as mentioned, analyzing them has become increasingly challenging. As a result, the data-stream model has emerged as a convenient model for how these data-sets are accessed in practice. In this model, the input matrix $A \in \R^{m \times n}$ is presented as a sequence of items/updates.
In one common setting, the \emph{turnstile} model, each update has the form $(i, j, \delta)$ for $\delta \in \Z$ and represents an update $A_{ij} \leftarrow A_{ij} + \delta$.
In another common setting, the \emph{row-order} model, items $(i,j,A_{ij})$ arrive in a fixed order, sorted by location $(i,j)$ lexicographically,  providing directly the entry of $A$ in that location. In both models, unspecified entries are $0$ by convention,
which is very effective for sparse matrices.

Row-order is a common access pattern for external memory algorithms.
When the data is too large to fit into working memory and has be ``streamed''
into memory in some pattern, it is useful to assume that algorithms can make multiple, albeit few, passes over the input data.
For a thorough discussion of such external memory algorithms,
including motivation for the row-order model and for multiple passes over the data, 
see~\cite{GM99,Vitter01,Lib13}. 

Designing small-space algorithms for estimating Schatten norms
of an input matrix in the data-stream model is an important problem,
and was investigated recently for various matrix classes and stream types
\cite{CW09,AN13,LNW14,LW16a,LW16b,LW17,BCKLWY18}.
However, all known algorithms require space that is polynomial in $n$,
the matrix dimension, 
even if the matrix is highly sparse and the stream type is favorable,
say row-order. A natural question then is: 
\begin{quote} 
{Does any streaming model admit algorithms for computing Schatten norms of a matrix presented as a stream, with storage requirement independent of the matrix dimension?}
\end{quote}
We \emph{answer this question in the affirmative}
for $k$-sparse matrices presented in row-order and all even integers $p$.
Our algorithms extend to all integers $p\geq 1$ 
if the input matrix is PSD. 

\subsection{Main Results}

Throughout, we write $\tilde O(f)$ as a shorthand for $O(f\cdot \log^{O(1)}n)$ where $n$ is the dimension of the matrix,
and write $O_d(f)$ when the hidden constant might depend on the parameter $d$.
We assume that the entries of the matrix are integers bounded by $\poly(n)$,
and thus often count space in words, each having $O(\log n)$ bits.
We denote by $\lceil p \rceil_4$ the smallest multiple of $4$ that is greater than or equal to $p$, and similarly by $\lfloor p \rfloor_4$ the largest multiple of $4$ that is smaller than or equal to $p$.

\paragraph{Upper and Lower Bounds for Row-Order Streams.}
Our main result is a new algorithm for approximating the Schatten $p$-norm
(for even $p$) of a $k$-sparse matrix streamed in row-order,
using $O(p)$ passes and $\poly(k^p/\epsilon)$ space
(independent of the matrix dimensions).
This is stated in the next theorem, 
whose proof appears in Section \ref{sec:Row-orderImportanceSampling}.

\begin{restatable}{theorem}{rowOrderAlgGeneralThm}
\label{thm:row-orderAlgGeneral}
There exists an algorithm that, given $p \in 2\Z_{\geq 2}$, $\epsilon > 0$ and a $k$-sparse matrix $A \in \R^{n \times n}$ streamed in row-order, makes $\floor{p/4}+1$ passes over the stream using $O_p(\epsilon^{-2}k^{3p/2 - 3})$ words of space, and outputs $\bar{Y}(A)$ that $(1\pm\epsilon)$-approximates $\spn{A}{p}$ with probability at least $2/3$.
\end{restatable}

Theorem \ref{thm:row-orderAlgGeneral} provides a multi-pass algorithm whose space complexity depends only on the sparsity of the input matrix.
A natural question is whether one can achieve a similar dependence
also for \emph{one-pass} algorithms,
but our next theorem (proved in Section~\ref{sec:row-orderlowerbounds})
shows that such algorithms require
$\Omega((n^{1-4/\lfloor p \rfloor_4})$ bits of space,
even for $O(1)$-sparse matrices. 

It follows that multiple passes over the data are necessary
for an algorithm for sparse matrices
to have space complexity independent of the matrix dimensions.

\begin{restatable}{theorem}{singlePassLBThm}\label{thm:singlePassLB}
For every $p \in 2\Z_{\geq 2}$ there is $\epsilon(p) > 0$ such that
every algorithm that makes one pass over an $O_p(1)$-sparse matrix $A \in \R^{n \times n}$ streamed in row-order,
and then outputs a $(1 \pm \epsilon(p))$-approximation to $\spn{A}{p}$ with probability at least $2/3$,
must use $\Omega(n^{1-4/\lfloor p \rfloor_4})$ bits of space. 
\end{restatable}

We can further extend our primary algorithmic technique
(from Theorem~\ref{thm:row-orderAlgGeneral})
in several different ways,
and obtain improved algorithms for special families of matrices,
algorithms in the more general turnstile model, 
and algorithms with a trade-off between the number of passes
and the space requirement,
as explaind later in this section. 
Table~\ref{tbl:results} summarizes our results for row-order streams,
and compares them to bounds derived from previous work (when applicable). 

\begin{table}[t]
\begin{center}
  \caption{Bounds for Schatten norms (for even $p$) of $k$-sparse matrices in row-order streams.
    Upper bound space is counted in words.
    Lower bounds are for suitable $\epsilon(p)>0$. 
    \label{tbl:results}
  }
\begin{tabular}{ |c|l l l| }
\hline
  \textbf{Which \boldmath $p$} &\textbf{Space Bound} &\textbf{Ref.} &\textbf{Comments} \\
  \hline\hline
  &&&\\
\multirow{5}{*}{$p > 4$} & $\tilde{O}_{p, \epsilon}(k^{O(p)}n^{1 - 4/\lceil p \rceil_4})$ & \cite{BCKLWY18} &one-pass\\
							& $O_{p, \epsilon}(k^{3p/2 - 3})$ &Thm.~\ref{thm:row-orderAlgGeneral} &$\floor{p/4}+1$ passes \\
              & $O_{p, \epsilon}(k^{2ps}n^{1 - 1/s})$ &Thm.~\ref{thm:sSetSchattenAlg} &$\lfloor \frac{p}{2(s+1)} \rfloor+1$ passes \\
              & & & \\
							& $\Omega (n^{1-{4}/{\lfloor p \rfloor_4}})$ &Thm.~\ref{thm:singlePassLB} &one-pass, $k=O(1)$ \\
							& $\Omega_t(k^{p/2 - 2})$ &  \cite{BCKLWY18} & $t$ passes, $k \leq n^{2/p}$\\
\hline &&&\\
\multirow{2}{*}{$p = 4$} 	& $\tilde{O}_{p, \epsilon}(k)$ &\cite{BCKLWY18} &one-pass\\
							& $O_{p}(\varepsilon^{-2})$ & Thm.~\ref{Thm: Schatten4analysis} &one-pass, for all $k \leq n$ \\
\hline
\end{tabular}

\end{center}
\hrule
\end{table}

\paragraph{Applications for Approximating Schatten Norms.}
We show in Section~\ref{sec:Applications} two settings where,
under certain simplifying conditions,
our algorithms can be used to approximate other functions of the spectrum, and even weakly recover the entire spectrum.
The basic idea is that it suffices to compute only a few Schatten norms,
in which case our algorithms for $k$-sparse matrices in row-order streams
can be used,
and the overall algorithm will require only small space (depending on $k$).

The first setting considers spectral sums of PSD matrices. 
We use an idea from~\cite{BDKKZ17} to show that for a PSD input $A \in \R^{n \times n}$ whose eigenvalues lie in an interval $[\theta, 1)$, one can $(1 \pm 2\epsilon)$-approximate $\log \det (A)$ and $\tr(A^{-1})$ using the first $\frac{1}{\theta}\log \left( \frac{1}{\epsilon}\right)$ (integer) Schatten norms. We further show that given a Laplacian matrix whose eigenvalues lie in an interval $[0,\theta]$, one can $(1 \pm 2\epsilon)$-approximate the Estrada index using the first $(e\theta+1)\log \frac{1}{\varepsilon}$ (integer) Schatten norms.

The second setting considers recovering the spectrum of a PSD matrix using a few Schatten norms of the matrix. We use an idea from~\cite{WV16} to  approximate the spectrum of a PSD matrix whose eigenvalues lie in the interval $[0,1]$, up to $L_1$-distance $\epsilon n$ using the first $O(1/\epsilon)$ Schatten norms.

\paragraph{Experiments.} 
We validated our row-order algorithm on real-world matrices representing academic collaboration network graphs. 
The experiments show that the space needed to approximate the Schatten $6$-norm of these matrices is much smaller than the theoretical bound,
and that the algorithm is efficient also for larger $p$ values.
In fact, the matrices in our experiments have $O(1)$-sparse in every row,
but their columns are only sparse on average. 
We also experimented to check if the algorithm is robust to noise,
and found that it is indeed effective also for nearly-sparse matrices. 
Our experiments validate that the storage requirement is independent of the matrix dimensions.
See Section~\ref{sec:experiments} for details.

\subsection{Extensions of Main Results}

\paragraph{Extension I: Fewer Passes.}
We show in Section~\ref{sec:passSpaceTrade-off} how to generalize our algorithmic technique to use fewer passes over the stream, albeit requiring more space.
Our method attains the following pass-space trade-off. 
For any integer $s\ge 2$,
our algorithm in Theorem~\ref{thm:sSetSchattenAlg} makes
$t(s) = \floor{\frac{p}{2(s+1)}}+1$
passes over the stream using
$O_{p}\left(\epsilon^{-3} k^{2ps} n^{1-1/s}\right)$ words of space,
and outputs a $(1 \pm \epsilon)$-approximation to $\spn{A}{p}$
for $p \in 2\Z_{\geq 2}$.

\paragraph{Extension II: Turnstile Streams.}
We design in Section~\ref{sec:TurnstileImportanceSampling} an algorithm for \emph{turnstile} streams with an additional $\tilde{O}(\epsilon^{-O(p)}k^{3p/2 - 3}n^{1 - 2/p})$ factor in their space complexity compared to our algorithm for row-order streams.
An additional $O(n^{1 - 2/p})$ factor is to be expected since the space complexity for estimating $\ell_p$ norms of vectors in turnstile streams is $\Omega(\frac{n^{1 - 2/p}}{t})$ if the algorithm is allowed to make $t$ passes over the data.
Our algorithm for turnstile streams makes $p + 1$ passes over the stream.
The algorithm of~\cite{LW16a} for $O(1)$-sparse matrices in the turnstile model
can obviously be extended to $k$-sparse matrices.
Its space requirement is $k^{O(p)}$, 
and we believe that a straightforward extension of their analysis
yields an exponent greater than $4.75p$

\paragraph{Extension III: Special Matrix Families.}
We give in Section~\ref{sec:Row-orderImportanceSampling} improved bounds for special families of $k$-sparse matrices that may be of potential interest.
We show that for Laplacians of undirected graphs with degree at most $k \in \N$, one can $(1 \pm \epsilon)$-approximate the Schatten $p$-norm with space ${O}_p(\epsilon^{-2}k^{p/2-1})$ by making $p/2$ passes over a row-order stream. Additionally, for matrices whose non-zero entries lie in an interval $[\alpha, \beta]$ for $\alpha, \beta \in \R^+$, we can get nearly-tight upper bounds -- our algorithm uses space $O_p(\epsilon^{-2}k^{p/2 - 1}(\beta/\alpha)^{p/2-2})$, which is nearly tight compared to the $\Omega(k^{p/2 - 2})$ multi-pass lower bound given in \cite{BCKLWY18} where $\alpha=\beta=1$.

\paragraph{Schatten 4-norm.}
We show in Section~\ref{sec:row-orderSchatten4} a simple one-pass algorithm for $(1 \pm \epsilon)$-approximating the Schatten $4$-norm of \emph{any} matrix
(not necessarily sparse) given in a row-order stream,
using only $\tilde{O}_p(\epsilon^{-2})$ words of space.
This improves a previous $\tilde{O}_p(\epsilon^{-2}k)$ bound from~\cite{BCKLWY18}.

\subsection{Technical Overview}\label{sec:TechnicalOverview}

\paragraph{Upper Bounds.}
We design an estimator that is inspired by the importance sampling framework and uses multiple passes over the data to implement the estimator. To the best of the our knowledge, this is the first algorithm for computing the Schatten $p$-norm in data streams that uses an adaptive sampling approach, i.e. the probability space of the algorithm's sampling in a given pass of the data is affected by the algorithm's decisions in the previous pass.

For an integer $p \in 2 \Z_{\geq 1}$ and $q \eqdef p/2$, the Schatten $p$-norm for a matrix $A \in \R^{n \times n}$, denoted by $\|A\|_{S_p}^p$, can be expressed as
\begin{equation}
\|A\|_{S_p}^p = \tr((A A^\top)^{q}) = \sum_{i_1, \ldots, i_{q} \in [n]} \ip{i_1}{i_2}\ip{i_2}{i_3}\ldots \ip{i_{q}}{i_1} \label{eqn:Schatten-p-original}
\end{equation}
where $a_i$ is the $i^{\text{th}}$ row of matrix $A$.

The Schatten $p$-norm can be interpreted using~\eqref{eqn:Schatten-p-original} as a sum over cycles of $q$ inner-products (which we refer to informally as \emph{cycles}) between rows of $A$. We assign each cycle in the above expression to one of the rows participating in that cycle. Hence, the Schatten $p$-norm can be expressed as a sum $\sum_{i = 1}^n z_i$ where $z_i$ is the cumulative weight of all the cycles assigned to row $i$.

Our algorithm starts by sampling a row $i \in [n]$ with probability proportional to the \emph{heaviest} cycle assigned to row $i$ (i.e., largest contribution to $z_i$). In the following $p/4$ stages, it samples one cycle assigned to $i$ with probability proportional to the weight of the cycle. Since the rows \emph{and} columns are sparse, each row cannot participate in ``too many'' cycles (because it is orthogonal to any row with a disjoint support). Specifically, we show that the number of cycles assigned to each row $i$ is only a function of $k$ and $p$. It follows that sampling the first row with probability proportional to the heaviest contributing cycle is a good approximation (up to a factor depending only on $k$ and $p$) to $z_i$, the actual contribution of row $i$ to $\sum_{i \in [n]}z_i = \spn{A}{p}$.

The space complexity of sampling a row with probability proportional to its heaviest contributing cycle depends on the assigning process. A natural assigning is to assign every cycle to the row with largest $l_2$-norm participating in that cycle (breaking ties arbitrarily). Notice then that, by the Cauchy-Schwarz inequality, the heaviest contributing cycle to row $i$ is simply $\|a_i\|_2^p$.

This estimator can be implemented in the row-order model easily by using weighted reservoir sampling \cite{Vit85, BOV15}, as shown in Section~\ref{sec:Row-orderImportanceSampling}. However, implementing it in turnstile streams is more challenging (see Section~\ref{sec:TurnstileImportanceSampling}). Using approximate $L_p$-samplers presented in \cite{MW10}, we build an approximate cascaded $L_{p, 2}$-norm\footnote{The $L_{p, 2}$-norm of a matrix $A \in \R^{n \times m}$ for $p \geq 0$ is $\left( \sum_{i = 1}^n \|a_i\|_2^p \right)^{1/p}$.} sampler, to sample rows $i$ with probability proportional to $\|a_i\|_2^p$. Additionally, we use the Count-Sketch data structure to recover rows and sample cycles once we have sampled the first, ``seed'' row. This allows us to implement the estimator in turnstile data streams with an additional $\tilde{O}(\epsilon^{-O(p)}n^{1 - 2/p})$ factor in the space complexity attributed to the approximate cascaded $L_{p, 2}$-norm sampler and an additional $O_p(k^{3p/2 - 3})$ factor that comes from approximating the sampling probabilities (compared to the row-order in which the sampling probabilities can be recovered exactly).

In Section \ref{sec:passSpaceTrade-off} we generalize the design of the importance sampling estimator. Instead of assigning every cycle to a single row that appears in it, every cycle is mapped to $s$ rows that participate in it, for parameter $s \in \N$. These $s$ rows split the cycle into roughly $\frac{q}{s}$ segments such that each of these $s$ rows participates in a segment where it is the heaviest' row (by $l_2$-norm). The algorithm samples $s$ ``seed" rows and then computes all the cycles (or alternatively samples one cycle) that are assigned to these $s$ rows. Since the length of each of the segments reduces linearly with $s$, one can compute these cycles with fewer passes. However, the algorithm needs to sample more indices in order to ensure that each cycle has a sufficiently large probability of being ``hit''. This tension leads to a trade-off between passes and space. 

\paragraph{Lower Bounds.}
We obtain an $\Omega(n^{1-4/{\lfloor p \rfloor_4}})$ bits lower bound for any algorithm that estimates the Schatten $p$-norm in one-pass of the stream for even $p$ values.
Our proof analyzes for even $p$ values a construction presented in \cite{LW16a}, which is based on a reduction from the Boolean Hidden Hypermatching problem. This lower bound holds even if the input matrix is promised to be $O(1)$-sparse.

\subsection{Previous and Related Work} \label{sec:previousWork}
The bilinear sketching algorithm in \cite{LNW14} was the first non-trivial algorithm for Schatten $p$-norm estimation in turnstile streams. It requires only one-pass over the data and uses $O(\epsilon^{-2}n^{2 - 4/p})$ words of space.\footnote{They also showed a lower bound of $\Omega(n^{2 - 4/p})$ for the dimension of bilinear sketching for approximating $\spn{A}{p}$ for all $p \geq 2$.} Their algorithm uses $O(\varepsilon^{-2})$ independent $G_1AG_2^\top$ sketches, where $G_1, G_2 \in \R^{t \times n}$ are matrices with i.i.d.~Gaussian entries and $t = O(n^{1 - 2/p})$.

Inspired by this sketch, \cite{BCKLWY18} gave an almost quadratic improvement in the space complexity if the algorithm is allowed to make multiple passes over the data. Their estimator uses matrices $G_2,\ldots, G_{p} \in \R^{t \times n}$ with i.i.d.~Gaussian entries and Gaussian vector $g_1 \in \R^n$ to output $g_1^\top AG_2^\top G_2A\ldots G_pAg_1$. This estimate can be constructed in $p/2$ passes of the data and requires $O(\varepsilon^{-2})$ independent copies each using only $t = O(n^{1 - \frac{1}{p-1}})$ space.

Restricting the input matrix to be $O(1)$-sparse allows for quadratic improvement in the space complexity of one-pass algorithms as shown in \cite{LW16a}. They show that sampling $O(n^{1 - 2/p})$ rows and storing them approximately using small space (since each row is sparse) is sufficient to $(1 + \epsilon)$-approximate the Schatten $p$-norm by exploiting the fact that rows cannot ``interact'' with one another ``too much'' because of the sparsity restriction.

If we restrict the data stream to be row-order, then we can reduce the dependence on $p$ in all the above algorithms by a factor of $2$. As noted in \cite{BCKLWY18}, since $A^\top A = \sum_{i} a_ia_i^\top$ (where $a_i$ is the $i^{\text{th}}$ row of $A$) one can apply the above algorithms to $A^\top A$ instead of $A$ by updating it with the outer product of every row with itself. Since $\spn{A^\top A}{p/2}=\spn{A}{p}$ (for even $p$ values), the output is as desired and the dependence on $p$ reduces by a factor of $2$.

\paragraph{Lower Bounds.} 
Every $t$-pass algorithm designed for turnstile streams requires $\Omega(n^{1-2/p}/t)$ bits,
which follows by injecting the $F_p$-moment problem (see \cite{Gro09, Jay09}) into the diagonal elements.
Li and Woodruff  \cite{LW16a} showed that restricting the algorithm to a single pass over the turnstile stream,
leads to a lower bound $\Omega(n^{1-\varepsilon})$ bits for every fixed $\varepsilon >0$ and $p \notin 2\Z_{\geq 2}$,
even if the input matrix is $O(1)$-sparse.%
\footnote{
They also showed that for $p \in 2\Z_{\geq 2}$,
single-pass algorithms require $\Omega(n^{1-2/p})$ bits even if all non-zeros in the input matrix are constants.
}
Later \cite{BCKLWY18} proved that $\Omega(n^{1-\varepsilon})$ bits are required for $p \notin 2\Z_{\geq 2}$ even in row-order streams.
In addition, they showed (Theorem 5.4 in Arxiv version) that $t$ passes over row-order streams require space 
$\Omega(n^{1-4/p}/t)$ bits, however these matrices are actually $\Omega(n^{2/p})$-sparse
(and not $O(1)$-sparse as may be understood from Table 2 therein).
A simple adaptation of that result yields an $\Omega(k^{p/2-2}/t)$
space lower bound for $k$-sparse input matrices ($k\leq n^{2/p}$).

\section{Notation and Preliminaries}\label{sec:notationAndPrelim}

The following useful fact comparing the lengths of the rows of $A$ and its Schatten $p$-norm is proved in Appendix \ref{app:lengthLeqSchatten-p}.
\begin{fact}\label{fact:lengthLeqSchatten-p}
Let matrix $A \in \R^{n \times n}$ have rows $\set{a_i}_{i\in [n]}$ and let $t \geq 1$. Then $\sum_{i \in [n]} \|a_i\|_2^{2t} \leq \spn{A}{2t}.$
\end{fact}

\paragraph{Importance Sampling.}

Our main algorithmic technique is inspired by the importance sampling framework, as formulated by the following theorem, proved in Appendix \ref{App:ImpSamplingProof}. 
\begin{theorem}[Importance Sampling]\label{thm:ImportanceSampling}
Let $z = \sum_{i \in [n]} z_i \geq 0$ be a sum of $n$ reals. Let the random variable $\hat{Z}$ be an estimator computed by sampling a single index $i \in [n]$ according to the probability distribution given by $\{\tau_i\}_{i = 1}^n$ and setting $\hat{Z} \leftarrow \frac{z_i}{\tau_i}$. If for some parameter $\lambda \geq 1$, each $\tau_i \geq \frac{|z_i|}{\lambda \cdot z}$, then
$$\Exp{}{\hat{Z}} = z \text{ \ \ and \ \ } \Var(\hat{Z}) \leq (\lambda z)^2.$$
\end{theorem}

\paragraph{Families of Matrices.}

We define two families of matrices that are of special interest.
\begin{itemize}
\item Let $\mathcal{L}_n \subseteq \Z^{n \times n}$ be the family of Laplacian matrices of undirected graphs $G([n],E)$ with positive edge-weights $\{w_{uv} > 0: uv\in E\}$.

\item Given positive constants $\alpha \leq \beta$, let $\Cc^{m\times n}_{\alpha, \beta} \subseteq \R^{m \times n}$ be the family of matrices $C$ such that every entry $C_{i, j}$ is either zero or in the range $[\alpha, \beta]$. For the vector case (i.e. $n=1$) we may write $\Cc^{m}_{\alpha, \beta}$.
\end{itemize}

\section{An Estimator for Schatten $p$-Norm for $p \in 2\Z_{\geq 2}$}\label{sec:ImpSamplingEstimatorPrelim}
This section introduces our importance sampling estimator for Schatten $p$-norms.
We begin in Section \ref{sec:EstimatorPreliminaries} with manipulating expression \eqref{eqn:Schatten-p-original} for the Schatten $p$-norm by assigning every summand,
i.e. a cycle of $p/2$ inner products, to its heaviest participating row, see \eqref{eqn:Schatten-p-Scaled}.
We then use this new expression in Section \ref{sec:ImpSamplingEstimator} to give an importance sampling estimator. 
In Section \ref{sec:ProjLemmas} we prove several lemmas, referred to as projection lemmas, which are key to our analysis in Section \ref{sec:estimatorAnalysis}.

\subsection{Preliminaries}\label{sec:EstimatorPreliminaries}
Fix a matrix $A \in \R^{n\times n}$ and $p \in 2\Z_{\geq 2}$. For a row $a_i$, we define the set of its \emph{neighboring rows} $N(i)  \eqdef \{l \in [n] : \supp{a_i} \cap \supp{a_l} \neq \emptyset \}$. In addition, we denote the set of neighboring rows of $a_j$ that have smaller length than row $a_i$
$$N_{S}^i(j)  \eqdef \{l\in N(j): \ \|a_l\|_2 \leq \|a_i\|_2\}.$$
Building on this, we intorduce notation for certain ``paths'' of rows. Fixing some row indices $i, i_1 \in [n]$ and an integer $t \geq 2$, we then define
\begin{align*}
\Gamma(i_1, t) & \eqdef \left\{(i_1, \ldots, i_t) : \ i_2 \in N(i_1), \ldots, i_t \in N(i_{t-1}) \right\}, \\
\Gamma_S^i(i_1, t) & \eqdef \left\{(i_1, \ldots, i_t) : \ i_2 \in N_S^i(i_1), \ldots, i_t \in N_S^i(i_{t-1}) \right\}.
\end{align*}

We further define the weights of ``paths'' of inner products: given an integer $t \geq 2$ and indices $i_1, \ldots, i_t \in [n]$, let
$$\sigma(i_1, \ldots, i_t)  \eqdef \ip{i_1}{i_2}\ip{i_2}{i_3}\ldots\ip{i_{t-1}}{i_t}.$$

Recall from \eqref{eqn:Schatten-p-original} that the Schatten $p$-norm of $A \in \R^{n \times n}$ can be expressed in terms of the product of inner products of the rows of $A$. Using the above notation we manipulate it as follows.
\begin{align}
\|A\|^p_{S_p} &= \tr \left( (AA^\top)^q \right) = \sum_{i_1, \ldots, i_{q} \in [n]} \sigma(i_1, \ldots, i_q, i_1) \\
&= \sum_{i_1} \sum_{ \substack{(i_1, \ldots, i_{q-1}) \\ \in \Gamma(i_1, q-1)} } \sum_{i_q \in N(i_1)} \sigma(i_1, \ldots, i_q, i_1) \label{eqn:Schatten-p-neighbors}\\
&= \sum_{i_1} \sum_{\substack{(i_1, \ldots, i_{q-1}) \\ \in \Gamma_S^{i_1}(i_1, q - 1)}} \sum_{i_q \in N^{i_1}_S(i_1)} c(i_1, \ldots, i_{q})\sigma(i_1, \ldots, i_q, i_1) \label{eqn:Schatten-p-Scaled}
\end{align}

where $1 \leq c(i_1, \ldots, i_{q}) \leq q $ is the number of times the sequence $(i_1, \ldots, i_q, i_1)$ or a cyclic shift of the sequence appears in Equation \eqref{eqn:Schatten-p-neighbors}.

\subsection{The Estimator}\label{sec:ImpSamplingEstimator}

Our estimator is an importance sampling estimator for the quantity in~\eqref{eqn:Schatten-p-Scaled}. To define it, we need the following quantities:
\begin{align*}
\Sc &\eqdef \bigcup_{i \in [n]} \Gamma_S^{i}(i, q-1) \\
z_{(i_1, \ldots, i_{q-1})} &\eqdef \sum_{i_q \in N^{i_1}_S(i_1)} c(i_1, \ldots, i_{q})\sigma(i_1, \ldots, i_q)\ip{i_q}{i_1} && \forall (i_1, \ldots, i_{q-1}) \in \Sc \\
z &\eqdef \sum_{ (i_1, \ldots, i_{q-1}) \in \Sc } z_{(i_1, \ldots, i_{q-1})} = \spn{A}{p} &&\text{by Equation~\eqref{eqn:Schatten-p-Scaled}.}
\end{align*}

Our importance sampling estimator, for the sum $z$, samples quantities $z_{(i_1, \ldots, i_{q-1})}$ indexed by $(i_1, \ldots, i_{q-1}) \in \Sc$ in $q-1$ steps. In the first step, it samples row $i_1\in [n]$ with probability $\frac{\|a_{i_1}\|_2^p}{\sum_{j} \|a_j\|_2^p}$. In each step $2 \leq t \leq q-1$, conditioned on sampling $i_{t-1}$ in step $t-1$ it samples row $i_t\in N_S^{i_1}(i_{t-1})$ with probability $$p^{i_1}_{i_{t - 1}} (i_t) \eqdef \frac{|\ip{i_{t-1}}{i_t}|}{\sum_{l \in N_S^{i_1}(i_{t-1})}|\ip{i_{t-1}}{l}|}.$$

Overall, a sequence $(i_1, \ldots, i_{q-1}) \in \Sc$ is sampled with probability
\begin{align*}
\tau_{(i_1, \ldots, i_{q-1})} &= \frac{\|a_{i_1}\|_2^p}{\sum_{j} \|a_j\|_2^p} \prod_{t = 2}^{q - 1} p^{i_1}_{i_{t - 1}} (i_t),\\
\intertext{and the output estimator is}
Y(A) &\eqdef \frac{1}{\tau_{(i_1, \ldots, i_{q-1})}}  \cdot z_{(i_1, \ldots, i_{q-1})}.
\end{align*}

\subsection{Projection Lemmas}\label{sec:ProjLemmas}

To analyze the estimator $Y(A)$, 
we need a few lemmas, which we call projection lemmas, for sparse matrices.
We start with two lemmas for sparse matrices,
followed by two lemmas for more specialized cases.

\begin{lemma}\label{Lem:InnerProdByLength}
For every \textit{$k$-sparse} matrix $B \in \R^{n \times k}$ with rows $b_1, \ldots, b_n$ and vector $x \in \R^k$ such that $\|x\|_2 \geq \|b_i\|_2$ for all $i \in [n]$, we have that
$$
  \frac{\|Bx\|_1}{\|x\|_2^2}
  = \sum_{i = 1}^n \frac{|\langle x, b_i\rangle|}{\|x\|_2^2} \leq k\sqrt{k}.
$$
\end{lemma}
\begin{proof}
For a vector $y \in \R^k$ and $S \subseteq [k]$, let $y_{|S}$ to be the restriction of $y$ onto its indices corresponding to set $S$.

For all $i \in [n]$, by the Cauchy-Schwarz inequality, 
$\tuple{x, b_i}
  = \tuple{x_{|\supp{b_i}}, b_i}
  \le \|x_{|\supp{b_i}}\|_2\|b_i \|_2
  $.   
Hence,
\begin{align*}
  \sum_{i = 1}^n \frac{|\langle x, b_i\rangle|}{\|x\|_2^2}
  & \leq \sum_{i = 1}^n \frac{\|x_{|\supp{b_i}}\|_2 \|b_i\|_2 }{ \|x\|_2^2 }
    \leq \sum_{i = 1}^n \frac{\|x_{|\supp{b_i}}\|_2 }{ \|x\|_2 }
    \leq \sum_{i = 1}^n \frac{\|x_{|\supp{b_i}}\|_1} {\|x\|_2}
    \leq \frac{k \|x\|_1}{\|x\|_2} , 
\end{align*}
where the last inequality follows from the sparsity of $B$
(every column index is in $\supp{b_i}$ for at most $k$ of the rows $b_i$).
The lemma now follows by a simple application of the Cauchy-Schwarz inequality.
\end{proof}

We need another, similar, lemma in order to bound the variance.

\begin{lemma}\label{Lem:InnerProdSquaredByLength}
For every \textit{$k$-sparse} matrix $B \in \R^{n \times k}$ with rows $b_1, \ldots, b_n$ and a vector $x \in \R^k$ such that $\|x\|_2 \geq \|b_i\|_2$ for all $i \in [n]$, we have that
$$\frac{\|Bx\|^2_2}{\|x\|_2^4} = \sum_{i = 1}^n \frac{\langle x, b_i\rangle^2}{\|x\|_2^4} \leq k.$$
\end{lemma}
\begin{proof}
Following similar steps as that of Lemma~\ref{Lem:InnerProdByLength},
\begin{align*}
  \sum_{i = 1}^n \frac{\langle x, b_i\rangle ^2}{\|x\|_2^4}
  \leq \sum_{i = 1}^n \frac{\|x_{|\supp{b_i}}\|_2^2} {\|x\|_2^2}
  \leq k ,
\end{align*}
where again the last inequality follows from the sparsity of $B$. 
\end{proof}

The next two lemmas present bounds that improve over
Lemma~\ref{Lem:InnerProdByLength} in two special cases,
when the $k$-sparse matrix is a graph Laplacian, 
and when all its non-zero entries come from a bounded range.

\begin{lemma}\label{Lem:GraphInnerProdByLength}
Let $G=([n],E)$ be an undirected graph with positive edge weights $\{w_{uv}\}_{uv\in E}$. 
Let $k$ be its maximum (unweighted) degree, and let $L(G) \in \R^{n \times n}$
be its Laplacian matrix with rows $l_1, \ldots, l_n$. 
Given $u \in [n]$, let the matrix $B_u$
consist of all the rows $l_v$ where $\|l_u\|_2 \geq \|l_v\|_2$,
and interpret $B_u$ also as a set of rows. 
Then,
$$
  \frac{\|B_u l_u\|_1}{\|l_u\|_2^2}
  = \sum_{l_v \in B_u} \frac{|\langle l_u, l_v\rangle|}{\|l_u\|_2^2} \leq 2k.
  $$
(Trivially, we can also omit from $B_u$ rows where $\langle l_u, l_v\rangle = 0$.) 
\end{lemma}

\begin{proof}
The main idea is that the additional matrix structure implies 
$
  \|l_u\|_1
  \le 2 \|l_u\|_2
  $,
which is better than what follows from the Cauchy-Schwarz inequality.
Indeed,
$
  \|l_u\|_2^2
  = \big(- \sum_{t\in N(u)}w_{ut}\big)^2 + \sum_{t\in N(u)}w_{ut}^2 
  \geq \big( \sum_{t\in N(u)}w_{ut} \big)^2
  = \big( \frac12 \|l_u\|_1 \big)^2
$. 
Now using this inequality in the proof of Lemma~\ref{Lem:InnerProdByLength},
we have 
$$\frac{\|B_u l_u\|_1}{\|l_u\|_2^2}\leq \frac{k \|l_u\|_1}{\|l_u\|_2}\leq 2k.$$
\end{proof}

\begin{lemma}\label{Lem:ConstantsInnerProdByLength}
For positive constants $\alpha \leq \beta$ and a \textit{$k$-sparse} matrix $B \in \Cc^{n\times k}_{\alpha, \beta}$ with rows $b_1, \ldots, b_n$ and a vector $x \in \Cc^{k}_{\alpha, \beta}$ such that $\|x\|_2 \geq \|b_i\|_2$ for all $i \in [n]$, we have that
$$\frac{\|Bx\|_1}{\|x\|_2^2} = \sum_{i = 1}^n \frac{|\langle x, b_i\rangle|}{\|x\|_2^2} \leq k \frac{\beta}{\alpha}.$$
\end{lemma}

\begin{proof}
By a direct calculation using the sparsity of $B$, 
\begin{align*}
  \sum_{i = 1}^n \frac{|\langle x, b_i\rangle|}{\|x\|_2^2}
  \leq \sum_{j=1}^k \frac{ \abs{x_j}\cdot \beta k }{ \alpha \|x\|_1 } 
  =  k \frac{\beta}{\alpha} .
\end{align*}
\end{proof}

\subsection{Analyzing the Estimator}\label{sec:estimatorAnalysis}
We now prove that the importance sampling estimator $Y(A)$ given in
Section~\ref{sec:ImpSamplingEstimator} is an unbiased estimator with a small variance.
In addition to analyzing the estimator for all $k$-sparse matrices,
we provide in Theorem~\ref{Thm:SpecialMatrices} improved bounds
for two special families of $k$-sparse matrices:
(i) Laplacians of undirected graphs and (ii) matrices whose non-zero entries lie in an interval $[\alpha, \beta]$ for parameters $0 < \alpha \leq \beta$.

\begin{theorem}\label{Thm:GeneralMatrices}
For every $p \in 2 \mathbb{Z}_{\geq 2}$ and a $k$-sparse matrix $A \in \R^{n \times n}$,
the estimator $Y(A)$ given in Section~\ref{sec:ImpSamplingEstimator} satisfies
$\Exp{}{Y(A)} = \|A\|_{S_p}^p$ and
$\Var(Y(A)) \leq O_p(k^{\frac{3p}{2} - 4}) \|A\|_{S_p}^{2p}$.
\end{theorem}

\begin{proof}
We will use the importance sampling framework of Theorem \ref{thm:ImportanceSampling}. In order to do so we must first argue that the values $\tau_{(i_1, \ldots, i_{q-1})}$ for $(i_1, \ldots, i_{q-1}) \in \Sc$ indeed form a probability distribution. It is easy to see that the probabilities of sampling the first row form a distribution over $[n]$. Similarly, for every $2 \leq t \leq q - 1$, the values $p^{i_1}_{i_{t-1}}( \cdot )$ indeed form a probability distribution over the rows in $N_S^{i_1}(i_{t-1})$. The argument for $\tau_{(i_1, \ldots, i_{q-1})}$ follows by the law of total probability.

Per Theorem \ref{thm:ImportanceSampling}, it is sufficient to prove that for all $(i_1, \ldots, i_{q-1}) \in \Sc$,
\begin{align}\label{Exp:lambdaEstimator}
\frac{1}{\tau_{(i_1, \ldots, i_{q-1})}}  \cdot \left| z_{(i_1, \ldots, i_{q-1})} \right| \leq O_p(k^{\frac{3}{4}p - 2}) z
\end{align}

Fix a sequence of indices $(i_1, \ldots, i_{q-1}) \in \Sc$. Inequality \eqref{Exp:lambdaEstimator} can be shown as follows,
\begin{align*}
\frac{\left| z_{(i_1, \ldots, i_{q-1})} \right|}{\tau_{(i_1, \ldots, i_{q-1})}} &= \frac{\sum_{j} \|a_j\|_2^p}{\|a_{i_1}\|_2^p} \prod_{t = 2}^{q - 1} \frac{1}{p^{i_1}_{i_{t - 1}} (i_t)} \left| \sum_{ i_{q} \in N^{i_1}_S(i_1) } c(i_1, \ldots, i_q)\sigma(i_1, \ldots, i_q)\ip{i_q}{i_1} \right| \\
&\leq  \frac{\sum_{j}\|a_j\|_2^p}{\|a_{i_1}\|_2^p} \ \frac{\prod_{t = 2}^{q-1} \sum_{l \in N_S^{i_1}(i_{t-1})}|\ip{i_{t-1}}{l}|}{|\sigma(i_1, \ldots, i_{q-1})|} \sum_{ i_{q} \in N^{i_1}_S(i_1) } c(i_1, \ldots, i_q)\left|\sigma(i_1, \ldots, i_q)\ip{i_q}{i_1}\right|  \\
&=  \frac{\sum_{j}\|a_j\|_2^p}{\|a_{i_1}\|_2^p}  \left( \prod_{t = 2}^{q-1} \sum_{l \in N_S^{i_1}(i_{t-1})}|\ip{i_{t-1}}{l}|   \right)  \sum_{ i_q \in N^{i_1}_S(i_1) } c(i_1, \ldots, i_q)\left|\ip{i_{q-1}}{i_q}\ip{i_q}{i_1}\right| \\
\intertext{By Young's Inequality for products of numbers and the bound on $c(i_1, \ldots, i_q)$,}
&\leq  \frac{q}{2}\frac{\sum_{j}\|a_j\|_2^p}{\|a_{i_1}\|_2^p}  \left( \prod_{t = 2}^{q-1} \sum_{l \in N_S^{i_1}(i_{t-1})}|\ip{i_{t-1}}{l}|   \right)  \left( \sum_{ i_q \in N^{i_1}_S(i_1) } \ip{i_{q-1}}{i_q}^2 +  \ip{i_q}{i_1}^2 \right)\\
&=  \frac{q}{2}\sum_{j}\|a_j\|_2^p  \left(\frac{\prod_{t = 2}^{q-1} \sum_{l \in N_S^{i_1}(i_{t-1})}|\ip{i_{t-1}}{l}|}{\|a_{i_1}\|_2^{p - 4}}   \right)  \left( \sum_{ i_q \in N^{i_1}_S(i_1) } \frac{\ip{i_{q-1}}{i_q}^2 + \ip{i_{q}}{i_1}^2}{\|a_{i_1}\|_2^4} \right)\\
\intertext{By applying Lemma \ref{Lem:InnerProdSquaredByLength} to the two inner-most summations and the fact that $\|a_{i_{q-1}}\|_2 \leq \|a_{i_1}\|_2$,}
&\leq q k \cdot \sum_{j}\|a_j\|_2^p  \left(\frac{\prod_{t = 2}^{q-1} \sum_{l \in N_S^{i_1}(i_{t-1})}|\ip{i_{t-1}}{l}|}{\|a_{i_1}\|_2^{p - 4}}   \right)  \\
\intertext{By applying Lemma \ref{Lem:InnerProdByLength} and the fact that $\|a_{i_{t-1}}\|_2 \leq \|a_{i_1}\|_2$ for any $2\leq t \leq q-1$, }
&\leq qk \sum_{j}\|a_j\|_2^p \left( \prod_{t = 2}^{q-1} k \sqrt{k} \right)
= qk^{\frac{3p}{4} - 2} \sum_{i}\|a_i\|_2^p \leq \frac{pk^{\frac{3p}{4} - 2}}{2}  \spn{A}{p}
\end{align*}
where the last inequality follows from Fact \ref{fact:lengthLeqSchatten-p}.
\end{proof}

\begin{theorem}\label{Thm:SpecialMatrices}
For every $p \in 2 \mathbb{Z}_{\geq 2}$ and a $k$-sparse Laplacian matrix $A \in \Lc_n$,
the estimator $Y(A)$ given in Section~\ref{sec:ImpSamplingEstimator}
satisfies $\Var(Y(A)) \leq O_p(k^{p/2-1}) \|A\|_{S_p}^{2p}$.
If instead the $k$-sparse matrix is $A\in \Cc^{n\times n}_{\alpha, \beta}$
for some $0 < \alpha \leq \beta$, then 
$\Var(Y(A))
  \leq O_p(k^{p/2 - 2}\left({\beta}/{\alpha} \right)^{p/2-2}) \|A\|_{S_p}^{2p}$.
\end{theorem}

\begin{proof}
The bound for $\Lc_n$ (Laplacians) follows the above proof of
Theorem~\ref{Thm:GeneralMatrices} but bounding the summations
using Lemma~\ref{Lem:GraphInnerProdByLength} 
instead of Lemma~\ref{Lem:InnerProdByLength}. 

The bound for $\Cc^{n\times n}_{\alpha, \beta}$ uses a special case of the importance sampling lemma. Using the notation from Theorem~\ref{thm:ImportanceSampling}, if $z_i > 0$ for all $i \in [n]$ then one can bound the variance by $\lambda (z)^2$.
Using this, the proof follows the same argument as that of Theorem~\ref{Thm:GeneralMatrices}
but using Lemma~\ref{Lem:ConstantsInnerProdByLength}
to bound the summations bounded by Lemmas~\ref{Lem:InnerProdSquaredByLength}
and~\ref{Lem:InnerProdByLength}.
\end{proof}

\section{Implementing the Estimator: Row-Order and Turnstile Streams}\label{sec:Implement}

In this section we show how to implement the importance sampling estimator defined in Section \ref{sec:ImpSamplingEstimator} in two different streaming models, row-order and turnstile streams. We start by stating two theorems that bound the space complexity of implementing the estimator in row-order streams.
The first one is our main result from the Introduction,
and applies to all $k$-sparse matrices. 
The second theorem considers special families of $k$-sparse matrices.

\rowOrderAlgGeneralThm*

\begin{restatable}{theorem}{rowOrderAlgSpecialThm}
\label{thm:row-orderAlgSpecial}
There exists an algorithm that, given $p \in 2\Z_{\geq 2}$, $\epsilon > 0$, and a $k$-sparse matrix $A \in \Lc_n$ streamed in row-order, makes $\floor{p/4}+1$ passes over the stream using $O_p({\epsilon^{-2}}k^{p/2})$ words of space, and outputs $\bar{Y}(A)$  that  $(1 \pm \epsilon)$-approximates $\spn{A}{p}$ with probability at least $2/3$.
If instead the $k$-sparse matrix $A$ is from $\Cc^{n\times n}_{\alpha, \beta}$
for $0 < \alpha \leq \beta$, 
then the space bound is $O_p(\epsilon^{-2}k^{p/2 - 1}\left({\beta}/{\alpha} \right)^{p/2-2})$ words.

\end{restatable}

We also show that the estimator defined in Section \ref{sec:ImpSamplingEstimator} can be implemented in turnstile streams in $p/2 + 3$ passes over the stream.

\begin{restatable}{theorem}{turnstileAlgThm}\label{thm:turnstileAlg}
There exists an algorithm that, given $p \in 2\Z_{\geq 2}$, $\epsilon > 0$ and a $k$-sparse matrix $A \in \R^{n \times n}$ streamed in a turnstile fashion, makes $p/2 + 3$ passes over the stream using \newline ${O}_p(k^{3p-6}n^{1 - \frac{2}{p}}(\epsilon^{-1}\log n)^{O(p)})$ words of space, and outputs  $\bar{Y}(A)$ that $(1 \pm \epsilon)$-approximates $\spn{A}{p}$ with probability at least $2/3$.
\end{restatable}

\paragraph{Outline.} 
At a high level, the algorithms in all three theorems are similar,
and compute multiple copies of the estimator defined in Section~\ref{sec:ImpSamplingEstimator} in parallel and output their average (to reduce the variance). The algorithms differ in the number of copies, derived from Theorems~\ref{Thm:GeneralMatrices} and~\ref{Thm:SpecialMatrices}. Here, and in Sections~\ref{sec:Row-orderImportanceSampling} and~\ref{sec:TurnstileImportanceSampling}, we describe how to implement each estimator in $p/2$ stages, and in Section~\ref{sec:reducingPasses} we show how to reduce the number of stages to $\floor{p/4}+1$. The first stage samples and stores a ``seed'' row which we will denote by $a_{i_1}$. Each subsequent stage $1 < t < q$ stores two values: a row index $i_t$ (and row $a_{i_t}$ itself) and an interim estimate $Y_t  \eqdef \sigma(i_1, \ldots, i_t)$. The final stage $q$ computes and outputs  $\sum_{i_q \in N_S^{i_1}(i_1)} Y_{q-1} \cdot \ip{i_q}{i_1}c(i_1, \ldots, i_q)$, where $1 \leq c(i_1, \ldots, i_q) \leq q$ is as defined in \eqref{eqn:Schatten-p-Scaled}.

The estimator is relatively easy to implement in row-order streams using $p/2$ passes and $O_p(\epsilon^{-2}k^{3p/2 - 3})$ words of space as shown in Section \ref{sec:Row-orderImportanceSampling}. In turnstile streams however, the estimator is more difficult to implement. The technical roadblock is sampling the first, ``seed'' row $i_1 \in [n]$ with probability proportional to $\frac{\|a_{i_1}\|_2^p}{\sum_j \|a_j\|_2^p}$. We use approximate samplers for turnstile streams to get around this roadblock. For a vector $x \in \R^n$ updated in a turnstile fashion, one can sample an index $i$ with probability approximately $x_i^t/\|x\|_t^t$ for various $t \in [0, \infty)$. Such algorithms are called $L_t$-samplers and have been studied thoroughly, see e.g. \cite{CH19}. Approximate samplers introduce a multiplicative (relative) error and an additive error in the sampling probability, which need to be accounted for when analyzing the algorithm that uses the sampler.  

Thus, in order to sample rows proportional to the quantities we want, we build two subroutines in the turnstile model:
\begin{enumerate}
\item Cascaded $L_{p, 2}$-norm sampler for $A$, used to sample the seed row $i_1$ with probability approximately $\|a_{i_1}\|_2^p$. It runs in \emph{$2$-passes}, has relative error $O(\epsilon)$ and uses space $\tilde{O}_p(\epsilon^{-2}n^{1 - 2/p})$.
\item Compute inner products between a given row and its neighbors in space $\tilde{O}(k^2)$.
\end{enumerate}
Using the two subroutines we can implement the estimator in Section \ref{sec:ImpSamplingEstimator} in $p+1$ passes of the stream in space $O_p(k^{3p - 6}n^{1 -{2/p}}(\epsilon^{-1}\log n)^{O(p)})$. The additional $\tilde{O}(n^{1 - 2/p})$ space complexity factor is introduced by the approximate $L_{p, 2}$-sampler. We remark that this factor is actually unavoidable for algorithms that compute $\spn{A}{p}$ in the turnstile model, since there is an $\Omega(n^{1 - 2/p})$ lower bound for computing the $l_p$-norm of vectors in $\R^n$ (in turnstile streams), even if the algorithm is allowed multiple passes. The additional $O(k^{3p/2 - 3})$ factor in the space complexity for turnstile streams compared to row-order streams is due to the bias introduced in estimating the sampling probability of the first, ``seed'' row.

As mentioned earlier, a slightly improved version runs in $\floor{p/4}+1$ and $p/2+3$ passes for row-order and turnstile streams respectively, with the same space complexities (up to constant factors). The idea is to build two parallel paths from the same seed row and eventually ``stitch'' the two into one cycle.

\subsection{Row-Order Streams}\label{sec:Row-orderImportanceSampling}
In this section we show how to easily implement the estimator defined in Section \ref{sec:ImpSamplingEstimator} in $q = p/2$ passes over a row-order stream, i.e. a sligthly weaker version of Theorem \ref{thm:row-orderAlgGeneral}. As mentioned, in Section \ref{sec:reducingPasses} we explain how to reduce the number of passes to $\floor{p/4}+1$ using a small adjustment to the algorithm. Algorithm \ref{Alg:row-order}, computes multiple copies of the estimator in parallel using space $O(k)$ for each copy.

\begin{algorithm}[H]
\caption{Algorithm for Schatten $p$-Norm of $k$-Sparse Matrices for $p \in 2\Z_{\geq 2}$ in Row-Order Streams}\label{Alg:row-order}

\begin{algorithmic}[1]
\Statex \textbf{Input}: $A \in \R^{n \times n}$ streamed in row-order, $p \in 2\Z_{\geq 2}$, $\epsilon > 0$, $m \in \Z^+$.
\Statex
\Parallel{$m$} \Comment{Each copy is a ``walk''}
	\State $i_1, \ldots, i_q \leftarrow 0$, $Y_1, \ldots, Y_q \leftarrow 0$
	\Pass{$1$}
		\State sample \emph{one} row $i_1 \in [n]$ with probability $\frac{\|a_{i_1}\|_2^p}{\sum_{j} \|a_j\|_2^p}$ \Comment{Using Reservoir Sampling}
		\State $Y_1 \leftarrow \frac{\sum_{j} \|a_j\|_2^p}{\|a_{i}\|_2^p}$
	\EndPass
	\Pass{$2 \leq t \leq q-1$}	
		\State sample \emph{one} row $i_t \in [n]$ with probability $p^{i_1}_{i_{t-1}}(i)$ \Comment{As defined in Section \ref{sec:ImpSamplingEstimator}}
		\State $Y_t \leftarrow Y_{t-1} \cdot \frac{\ip{i_{t-1}}{i_t}}{p^{i_1}_{i_{t-1}}(i)}$
	\EndPass
	\Pass{$q$}
		\State compute $Y_q \leftarrow Y_{q-1}\sum_{i_q \in N^{i_1}_S(i_1)}\ip{i_{q-1}}{i_q}\ip{i_q}{i_1}c(i_1, \ldots, i_{q})$
	\EndPass
\EndParallel
\State \textbf{return} average of the $m$ copies of $Y_q$
\end{algorithmic}
\end{algorithm}

\begin{proof}[Proof of Theorem~\ref{thm:row-orderAlgGeneral} (version with $p/2$ passes)]
Algorithm \ref{Alg:row-order} computes the estimator defined in Section \ref{sec:ImpSamplingEstimator} $m$ times in parallel and outputs the average which we will denote by $\bar{Y}(A)$. Since the variance of the estimator is at most $C_pk^{\frac{3p}{2} - 4}$ as per Theorem \ref{Thm:GeneralMatrices}, by setting $m = \frac{Ck^{\frac{3p}{2} - 4}}{\epsilon^2}$ and the constant $C$ appropriately, the guarantee on the estimate follows by an application of Chebyshev's Inequality to $\bar{Y}(A)$.

In pass $t$, each instance of the $m$ parallel instances store the row $a_{i_t}$ along with other estimates that can be stored in $O_p(1)$ words of space. Thus the total space complexity of the algorithm is $mk = O_p(\epsilon^{-2}k^{\frac{3p}{2} - 3})$ words.
\end{proof}

The proof of Theorem~\ref{thm:row-orderAlgSpecial} (version with $p/2$ passes) follows the above by adjusting $m$ according to Theorem~\ref{Thm:SpecialMatrices}.

\subsection{Turnstile Streams}\label{sec:TurnstileImportanceSampling}

\subsubsection{Preliminaries for Approximate Sampling}
We define approximate samplers which we will use in turnstile streams to implement our estimator. Approximate $L_p$ samplers have been studied extensively, see e.g. \cite{CH19}.

\begin{definition}[Approximate $L_t$ Sampler]\label{def:Lpsampler}
Let $x \in \R^n$ be a vector and $t \geq 0$. An \emph{approximate $L_t$ sampler} with relative error $\epsilon$, additive error $\Delta$, and success probability $1 - \delta$ is an algorithm that outputs each index $i \in [n]$ with probability 
$$p_i \in (1 \pm \epsilon) \frac{|x_i|^t}{\|x\|_t^t} \pm \Delta,$$
and with probability $\delta$ the sampler is allowed to output $\mathsf{FAIL}$.
\end{definition}

If an approximate sampler has no relative error and its additive error is less than $n^{-C}$, for arbitrarily large constant $C > 0$, then it is referred to as an \emph{exact} $L_p$-sampler.

Generalizing $L_p$-samplers, we define approximate $L_{p, q}$-samplers for matrices.

\begin{definition}[Weak Approximate $L_{t, q}$ Sampler]\label{def:Lpqsampler}
Let $t, q \geq 0$ be constants and $A \in \R^{n \times m}$ be a matrix with rows $a_1, \ldots, a_n$. An \emph{approximate $L_{t, q}$ sampler} with relative error $\epsilon$, additive error $\Delta$, and success probability $1 - \delta$ is an algorithm that, conditioned on succeeding, outputs each index $i \in [n]$ with probability 
$$p_i \in (1 \pm \epsilon) \frac{\|a_i\|_q^t}{\sum_{j \in [n]} \|a_j\|_q^t} \pm \Delta,$$
and on failing, which occurs with probability $\delta$, outputs any index.
\end{definition}

We draw the attention of the reader to the success condition of the $L_{p, q}$ sampler; unlike for $L_p$ samplers, the above definition is a weaker guarantee but is sufficient for our purpose since we can absorb the probability of failure for the sampler into the failure probability of the Schatten $p$-norm algorithm.

We recall some properties of higher powers of Gaussian distributions which we will use later in the analysis of sampling subroutines that we build. First, we give the higher moments of mean zero Gaussian random variables.

\begin{fact}\label{fact:GaussianRthMomentExp}For $t \geq 0$, $r \in 2\Z_{\geq 1}$ and a random variable $X \sim \mathcal{N}(0, t^2)$, we have $$ \Exp{}{|X|^r} = t^{r}(r - 1)!!.$$
\end{fact}

We state a concentration property for polynomial functions of independent Gaussian/Rademacher random variables called Hypercontractivity Inequalities. For an introduction to the theory of hypercontractivity, see e.g. Chapter 9 of \cite{O'Donnell2014}.

\begin{proposition}[Hypercontractivity Concentration Inequality, Theorem 1.9 \cite{SS12}]\label{thm:hypercontractivityGaussian}
Consider a degree $d$ polynomial $f(Y) = f(Y_1, \ldots, Y_n)$ of independent centered Gaussian or Rademacher random variables $Y_1, \ldots, Y_n$. Denote the variance $\sigma^2 = \Var{(f(Y))}$, then,
$$  \forall \lambda \geq 0, \quad \Prob{\left| f(Y) - \Exp{}{f(Y)} \right| \geq \lambda} \leq e^2 \exp \left( -\left( \frac{\lambda^2}{R \cdot \sigma^2} \right)^\frac{1}{d} \right)$$
where $R = R(d) > 0$ depends only on $d$.
\end{proposition}

\subsubsection{Weak Sampler for Cascaded Norm $L_{p, 2}$}

Before giving our construction for approximate $L_{p, 2}$ samplers in the turnstile model (Theorem \ref{thm:ApproximateLp2sampling}), we recall some core results for $L_p$ samplers that will be the algorithmic workhorse of our subroutine for $L_{p, 2}$ sampling.

One can construct algorithms for approximate $L_p$ samplers in various computational models. We look specifically at $L_p$ samplers in the turnstile streaming model. The following algorithmic guarantees exist for approximate $L_p$ samplers of vectors in turnstile streams.

\begin{theorem}[Theorem 1.2 in \cite{MW10}]\label{thm:ApproximateLpsampling} For $\delta > 0$ and $p \in 2\Z^+$, there exists an $0$-relative-error $L_p$-sampler in turnstile streams, in \textbf{$2$-passes}, with probability of outputting \textsf{FAIL} at most $n^{-C}$ where $C > 0$ is an arbitrarily large constant. The algorithm uses $O_p(n^{1 - 2/p}\log^{O(p)}n)$ space. \footnote{The original theorem statement in the paper is for $p \in [0, 2]$ but it is well-known among experts that the result extends to $p > 2$.}
\end{theorem}

For a given vector $x \in \R^n$ whose entries are streamed in a turnstile fashion, we will denote \textsc{$L_p$-Sampler}$(x, \delta)$ to be the output of the algorithm in Theorem \ref{thm:ApproximateLpsampling} with failure probability at most $\delta$. We will use this algorithm in turnstile streams for $p \geq 2$ to give an $O(\epsilon)$ relative error $L_{p, 2}$ sampler and failure probability at most $\delta$ for any given $\delta > 0$. The algorithm is fairly simple and is described in Algorithm \ref{Alg:L2psampler}.

\begin{algorithm}[H]
\caption{Approximate $L_{p, 2}$ Sampling Algorithm in Turnstile Streams}\label{Alg:L2psampler}

\begin{algorithmic}[1]
\Statex \textbf{INPUT}: $A \in \R^{n \times n}$ as a turnstile stream, $p \in \Z_{\geq 2}$, $\hat{\delta} \in (0, 1), \epsilon > 0$.
\Statex
\State Set $\hat{C}_p > 0$, $m \leftarrow \frac{\hat{C}_p{\log^p n}}{\epsilon^2}$ \Comment{$\hat{C}_p$ depends only on $p$}
\State construct $G \in \R^{n \times m}$, with i.i.d standard Gaussian entries\Comment{drawn pseudorandomly}
\State compute matrix $X \leftarrow \frac{1}{(p - 1)!!} \cdot AG$
\State $(i, j) \leftarrow \textsc{$L_p$-Sampler}(x, \hat{\delta})$ \Comment{where $x \in \R^{n^2}$ is the ``flattened'' version of $X$}
\State \textbf{return} $i$ if $L_p$-sampler didn't output $\mathsf{FAIL}$ otherwise return any index
\end{algorithmic}
\end{algorithm}

The matrix $X$, defined on line 3 in the above algorithm, can be computed ``on the fly'' given updates to $A$ in the stream.  

We then give the following theorem for approximate $L_{p, 2}$ sampling in turnstile streams by arguing  for the vector $x$ defined in Algorithm \ref{Alg:L2psampler}, the average of the $p^{\text{th}}$ power of the entries corresponding to row $i$ is tightly concentrated around $\| a_i \|_2^p$.

\begin{theorem}\label{thm:ApproximateLp2sampling} For every $\epsilon, C > 0, \delta \in (0, 1)$ and $p \in 2\Z_{\geq 2}$, Algorithm \ref{Alg:L2psampler} is an $O(\epsilon)$ relative error and $O(n^{-C})$ additive error $L_{p, 2}$ weak sampler in turnstile streams with failure probability at most $\delta$. The algorithm uses $O_p( n^{1 - 2/p}\epsilon^{-2}\log(\frac{1}{\delta})\log^{O(p)}(n))$ words of space.
\end{theorem}

\begin{proof}
For a fixed $i \in [n]$, notice that $x_{i, 1}, \ldots, x_{i, m}$ are independent and identically distributed as $\mathcal{N} \left(0, \frac{\|a_i\|_2^2}{((p-1)!!)^2} \right)$. Using Fact \ref{fact:GaussianRthMomentExp}, $\Exp{}{x_{i, j}^p} = \|a_i\|_2^p$ for all $j \in [m]$ since $p$ is even.

Let $i^*$ be the output of Algorithm \ref{Alg:L2psampler}. From the guarantee for $L_{p}$-samplers by Theorem \ref{thm:ApproximateLpsampling}, conditioning on the $L_p$ sampler succeeding, and setting the additive error sufficiently low, the probability that $i^* = i$ is
$$\Prob{i^* = i} = \sum_{j = 1}^m\frac{x_{i, j}^p}{\|x\|_p^p} \pm O\left(n^{-C}\right).$$
We will first show that, for a fixed $i \in [n]$, the quantity $\sum_{j = 1}^m x_{i, j}^p$ is tightly concentrated around $m\|a_i\|_2^p$ with high probability \emph{over the randomness of the Gaussian sketch}.

Set the polynomial $f : \R^m \rightarrow \R$ on the random variables $\{x_{i, j}\}_{j = 1}^m$ to be $f(x_{i, 1}, \ldots, x_{i, m}) = \sum_{j = 1}^m x_{i, j}^m$. Since the random variables $\{x_{i, j}\}_{j = 1}^m$ are independent, $$\Var(f(x_{i, 1}, \ldots, x_{i, m})) = m\Var(x_{i, *}^p) = m\|a_i\|_2^{2p}\frac{(2p-1)!! - ((p-1)!!)^2}{((p-1)!!)^2}$$
for even $p > 2$. Using this to apply the Hypercontractivity Concentration Inequality for Gaussian random variables given in Proposition \ref{thm:hypercontractivityGaussian} gives us,
$$\Prob{ \left| \sum_{j = 1}^m x_{i, j}^p - m\|a_i\|_2^p \right| \geq \epsilon m\|a_i\|_2^p } \leq e^2 \exp \left(- \left( \frac{\epsilon^2 m}{C_p} \right)^{\frac{1}{p}} \right)$$ where $C_p$ is a constant only dependent on $p$.

By setting $\hat{C}_p$ in Algorithm \ref{Alg:L2psampler} appropriately, we can apply the the union bound over all $i \in [n]$ to obtain,
$$\Prob{i^* = i} = \frac{(1 \pm O(\epsilon))\|a_{i}\|_2^p}{(1 \pm O(\epsilon))\sum_{l = 1}^n \|a_l\|_2^p}  \pm O(n^{-C})\ \ \ \ \ \ \ \ \text{for all $i \in [n]$}$$ with probability at least $1 -
\hat{\delta} -n^{-\hat{c}}$ (where $\hat{c}$ is dependent on $\hat{C_p}$). Setting $\hat{\delta}$ appropriately in Algorithm \ref{Alg:L2psampler} gives us the theorem.
\end{proof}

\subsubsection{Recovering Rows and Their Neighbors}
\label{sec:NeighborsRecoveryTurnstile}
We also give some subroutines to recover rows and their neighbors so that we can compute inner-products between rows, sample neighbors and compute the probabilities for the estimator. The algorithmic core for these subroutines will be sparse-recovery algorithms which can be implemented using the Count-Sketch data structure described below.

\begin{theorem}[Count-Sketch \cite{CCF04}]\label{thm:Count-Sketch}
For all $w, n \in \Z^+$ and $\delta \in (0,1)$, there is a randomized linear function $M : \R^n \leftarrow \R^s$ with $S = O(w\log(n/\delta))$ and a recovery algorithm $A$ satisfying the following. For input $x \in \R^n$, algorithm $A$ reads $Mx$ and outputs a vector $\tilde{x} \in \R^n$ such that
$$\forall x \in \R^n, \quad \Prob{\|x - \tilde{x}\|_\infty \leq \frac{1}{\sqrt{w}} \min_{{x'} : \|x'\|_0 = w} \|x - x'\|_2} \geq 1 - \delta.$$
\end{theorem}

Denote the output of a Count-Sketch algorithm on vector $x \in \R^n$ with parameter $w \in \Z^+$ and failure probability $\delta \geq 0$ to be $\textsc{Count-Sketch}_w(x, \delta)$. Notice that if it is guaranteed that $x$ is $k$-sparse, i.e. $\|x\|_0 \leq k$, then the output $\textsc{Count-Sketch}_k(x, \delta)$ recovers the vector $x$ exactly with probability at least $1 - \delta$ because $\min_{\tilde{x} : \|\tilde{x}\|_0 = k} \|x - \tilde{x}\|_2 = 0$ for every $k$-sparse vector $x$.

Reverting to our setting of $k$-sparse matrices in turnstile streams, given a target index $i \in [n]$, it is clear how to recover row $a_i$ using $\tilde{O}(k)$ space using the Count-Sketch algorithm stated. Given a row $a_i$, we can recover the neighboring rows $\{a_j : j \in N(i)\}$ by running $\textsc{Count-Sketch}_k(A_{*, j}, \tilde{\delta})$ for each $j \in \supp{a_i}$ (where $A_{*, j}$ corresponds to the $j^{\text{th}}$ column of $A$). Since each column and row is $k$-sparse, with $\tilde{O}(k^2)$ space, we can recover the neighbors of row $a_i$ given access to $a_i$. In addition, by setting the failure probability to $\frac{\delta}{k+1}$ in the above calls to $\textsc{Count-Sketch}_k$, our recovery subroutine will succeed with probability at least $1 - \delta$.

\subsubsection{Algorithm for Turnstile Streams}

We are now ready to present the algorithm implementing the estimator stated in Section \ref{sec:ImpSamplingEstimator} for turnstile streams. We note that unlike in row-order streams, we cannot recover the probability of sampling the first row exactly in turnstile streams. Since the output probability of the samplers is approximate, it introduces some bias in the estimator which we will have to bound. Therefore, the proof of correctness for this algorithm slightly deviates from that given in Theorem \ref{thm:ImportanceSampling} but uses the same underlying ideas.

Let us introduce notation for the subroutines we will need. Denote by $\textsc{$L_{p, 2}$-Sampler}(A, \epsilon, \delta)$ the output of the approximate $L_{p, 2}$ sampler defined in Algorithm \ref{Alg:L2psampler} with relative error $\epsilon$, and failure probability $\delta$. Additionally, we will need to estimate the cascaded norm $L_{p, 2}$ of $A$ in order to bias the quantity we sample in our importance sampling estimator. Denote by $\textsc{$L_{p, 2}$-NormEstimator}(A, \epsilon, \delta)$ the output of an algorithm for estimating the $L_{p, 2}$-norm of $A$ with relative error $\epsilon$ and failure probability $\delta$, such as in Section 4 of \cite{JW09}.

We describe our algorithm for turnstile streams in Algorithm \ref{Alg:turnstile}, which runs $p+1$ passes over the data, i.e.  a sligthly weaker version of Theorem \ref{thm:turnstileAlg}. As mentioned, the number of passes can be reduced to $\floor{p/2}+3$ using the extra insight of Section \ref{sec:reducingPasses}. 

\begin{algorithm}[H]
\caption{Algorithm for Schatten $p$-norm of $k$-Sparse Matrices for $p \in 2\Z_{\geq 2}$ in Turnstile Streams}\label{Alg:turnstile}

\begin{algorithmic}[1]
\Statex \textbf{Input}: $A \in \R^{n \times n}$ in a stream with turnstile updates, $p \in 2\Z_{\geq 2}$, $\epsilon > 0$, $m \in \Z^+$.
\Statex
\Parallel{$m$}	
	\State $i_1, \ldots, i_q \leftarrow 0$, $Y_1, \ldots, Y_q \leftarrow 0$
	\Stage{$1$} \Comment{takes 3 passes} \label{step:stage1turnstile}
		\State $i_1 \leftarrow \textsc{$L_{p, 2}$-Sampler}(A, \frac{\epsilon}{k^{3p/4 - 2}}, \frac{1}{100})$
		\State $\tilde{a}_{i_1} \leftarrow \textsc{Count-Sketch}_k(a_{i_1}, \frac{1}{100})$
		\State $D_1 \leftarrow \textsc{$L_{p, 2}$-NormEstimator}(A, \epsilon, \frac{1}{100})$\label{step:l_p_2Sampler}
		\State $Y_1 \leftarrow \frac{D_1}{\|\tilde{a}_{i_1}\|_2^p}$
	\EndStage
	\Stage{$2 \leq t \leq q-1$}	\Comment{each stage takes 2 passes}
		\State $\tilde{C}_{t-1} \leftarrow \{\textsc{Count-Sketch}_k(A_{*, j}, \frac{1}{100kq}) : j \in \supp{\tilde{a}_{i_{t-1}}}\}$
		\State reconstruct rows $\tilde{R}_{t-1} \leftarrow \{r_j : \text{row $j$ has support in } \tilde{C}_{t-1} \text{ and has $l_2$-norm less than $\tilde{a}_{i_1}$}\}$.
		\State $D_t \leftarrow \sum_{j \in \tilde{R}_{t-1}} \langle \tilde{a}_{i_{t-1}}, r_j \rangle$
		\State sample row index $i_t \in \supp{\tilde{R}_{t-1}}$ with probability $\frac{\langle \tilde{a}_{i_{t-1}}, r_{i_t} \rangle}{D_t}$
		\State $\tilde{a}_{i_t} \leftarrow \textsc{Count-Sketch}_k(a_{i_t}, \frac{1}{100q})$
		\State $Y_t \leftarrow Y_{t-1} \cdot \frac{D_t}{\tildeip{i_{t-1}}{i_{t}}} \cdot \tildeip{i_{t-1}}{i_{t}}$
	\EndStage
	\Stage{$q$}
		\State $\tilde{C}_{q-1} \leftarrow \{\textsc{Count-Sketch}_k(A_{*, j}, \frac{1}{100k}) : j \in \supp{\tilde{a}_{i_{q-1}}}\}$	
		\State reconstruct rows $\tilde{R}_{q-1} \leftarrow \{r_j : \text{row $j$ has support in } \tilde{C}_{q-1} \text{ and has $l_2$-norm less than $\tilde{a}_{i_1}$}\}$.
		\State compute $$Y_q \leftarrow Y_{q-1}\sum_{r_j \in \tilde{R}_{q-1}}{\langle \tilde{a}_{i_{q-1}}, r_j \rangle}{\langle r_j, \tilde{a}_{i_1} \rangle}c(i_1, \ldots, i_{q-1}, j)$$
	\EndStage
\EndParallel
\State \textbf{return} average of the $m$ copies of $Y_q$
\end{algorithmic}
\end{algorithm}

\begin{proof}[Proof of Theorem~\ref{thm:turnstileAlg} (version with $p+1$ passes)]
Recall from Section~\ref{sec:NeighborsRecoveryTurnstile} that $\textsc{Count-Sketch}_k$ will recover all the entries of a $k$-sparse vector exactly with high probability. By setting the failure probability of each call to $\textsc{Count-Sketch}_k$ to be sufficiently low, we can apply a union bound over all executions and assume that the algorithm recovers all the rows denoted by $\tilde{a}$ and $r$.

Let us assume that the $L_p$-sampler and Count-Sketch routines succeed and argue that taking the expectation over the randomness of the Gaussian sketch in the $L_{p, 2}$-Sampler algorithm, the $\textsc{$L_{p, 2}$-NormEstimator}$ and the importance sampling estimator gives us that $|\Exp{}{\bar{Y}(A)} - \spn{A}{p}| \leq O_p(\epsilon)\spn{A}{p}$.

Recall that the algorithm invokes an $O \left( \frac{\epsilon}{k^{3p/4 - 2}} \right)$ relative error $L_{p, 2}$-sampler in line 4. Since the additive error is less than $n^{-C}$ for arbitrary $C \geq 0$, we can simply absorb it in the failure probability of the algorithm. We thus get,
\begin{align*}
\Exp{}{\bar{Y}(A)} &= \sum_{\substack{(i_1, \ldots, i_{q- 1}) \\ \in \Sc}} \left( 1 \pm \frac{O(\epsilon)}{k^{3p/4 - 2}} \right)  \frac{\|a_{i_1}\|_2^p}{\sum_{j} \|a_j\|_2^p}\frac{\Exp{}{D_1}}{\|a_{i_1}\|_2^p} \sum_{i_q \in N^{i_1}_S(i_1)} \sigma(i_1, \ldots, i_q, i_1)c(i_1, \ldots, i_q)  \\
\intertext{Since \textsc{$L_{p, 2}$-NormEstimator} is an unbiased estimator for the $L_{p, 2}$-norm, i.e. $\Exp{}{D_1} = \sum_j \|a_j\|_2^p$, we get}
\left | \Exp{}{\bar{Y}(A)} - \spn{A}{p} \right| &\leq \sum_{\substack{(i_1, \ldots, i_{q- 1}) \\ \in \Sc}} \frac{O(\epsilon)}{k^{3p/4 - 2}} \left|  \sum_{i_q \in N^{i_1}_S(i_1)}  \sigma(i_1, \ldots, i_q, i_1)c(i_1, \ldots, i_q) \right|
\end{align*}
We can upper bound the second term as we did in bounding the variance of the estimator in Theorem \ref{thm:ImportanceSampling} to get $\left|\Exp{}{\bar{Y}(A)} - \spn{A}{p} \right| \leq O_p(\epsilon)\spn{A}{p}$

It is left to bound the variance of $\bar{Y}(A)$. Again, we assume that the $L_p$-Sampler and Count-Sketch routines succeed and recall that that for a sequence $(i_1, \dots, i_{q-1}) \in \Sc$, we define $z_{(i_1, \dots, i_{q-1})} = \sum_{i_q \in N^{i_1}_S(i_1)} \sigma(i_1, \ldots, i_q, i_1)c(i_1, \ldots, i_q)$. Given the guarantee of $L_{p, 2}$ sampling in Theorem \ref{thm:ApproximateLp2sampling}, the variance of the estimate $\bar{Y}(A)$ is
\begin{align*}
\Var{(\bar{Y}(A))} &\leq \frac{1}{m}\sum_{\substack{(i_1, \ldots, i_{q- 1}) \\ \in \Sc}} (1 \pm \frac{O(\epsilon)}{k^{3p/4 - 2}})\frac{1}{\sum_{j} \|a_j\|_2^p}\frac{\Exp{}{D_1^2}}{\|a_{i_1}\|_2^p} \prod_{t = 2}^{q-1} \frac{1}{p_{i_{t-1}}^{i_1}(i_t)} \left(z_{(i_1, \dots, i_{q-1})} \right)^2 \\
\intertext{By the accuracy guarantee of \textsc{$L_{p, 2}$-NormEstimator} and Fact \ref{fact:lengthLeqSchatten-p},}
&\leq \frac{1}{m} \sum_{\substack{(i_1, \ldots, i_{q- 1}) \\ \in \Sc}} (1 \pm O(\epsilon))\frac{\spn{A}{p}}{\|a_{i_1}\|_2^p} \prod_{t = 2}^{q-1} \frac{1}{p_{i_{t-1}}^{i_1}(i_t)} \left(z_{(i_1, \dots, i_{q-1})} \right)^2
\end{align*}
Bounding this identically as we did in Theorem \ref{thm:ImportanceSampling} and setting $m = \frac{Ck^{3p/2- 4}}{\epsilon^2 }$ give us $\Var(\bar{Y}(A)) \leq C_p\epsilon\|A\|_{S_p}^{2p}$ where $C_p$ is a constant dependent only on $p$.

The $\textsc{$L_{p, 2}$-Sampler}$ with $O \left(\frac{\epsilon}{k^{3p/4 - 2}} \right)$ relative error takes space $\tilde{O}_p(k^{\frac{3p}{2} - 4}n^{1-\frac{2}{p}}(\epsilon^{-1}\log n)^{O(p)})$ and the $\textsc{$L_{p, 2}$-NormEstimator}$ takes space $\tilde{O}_p(n^{1-\frac{2}{p}}(\epsilon^{-1}\log n)^{O(p)})$. In addition, storing the rows recovered from Count-Sketch requires $\tilde{O}(k^2)$ space. Thus, the space complexity of repeating the estimator $m = \frac{Ck^{3p/2- 4}}{\epsilon^2}$ times is $\tilde{O}_p(k^{3p - 6}n^{1 - \frac{2}{p}}(\epsilon^{-1}\log n)^{O(p)})$. We note that in stage 1, the sampler takes two passes, followed by another pass for Count-Sketch and the norm estimator. The subsequent stages requires two passes each giving a total of $3 + 2(q - 1) = p + 1$ passes.
\end{proof}

\subsection{Fewer Passes}
\label{sec:reducingPasses}

As mentioned earlier, we can slightly modify the way we implement the estimator from Section~\ref{sec:ImpSamplingEstimator} to reduce the number of passes that Algorithm~\ref{Alg:row-order} and Algorithm~\ref{Alg:turnstile} make to $\floor{\frac{p}{4}}+1$ and $\frac{p}{2} + 3$, respectively.
This is explained below and completes the proofs of
Theorems~\ref{thm:row-orderAlgGeneral}, \ref{thm:row-orderAlgSpecial} and \ref{thm:turnstileAlg}. 

The idea is to sample each sequence $(i_1, \ldots, i_q) \in \Sc$ in a different way albeit with the same probability. Assume for simplicity that $p \equiv 0\pmod4$. After sampling the first row $i_1 \in [n]$, we sample independently two ``paths'' of length $p/4 - 1$, each starting at $i_1$, with probabilities identical to the ones in the estimator. We then sum over the common neighbors of the endpoints of the two paths, using each of them to complete a cycle of length $p/2$. Formally, sample independently two sequences of rows $(i_1,l_1,\ldots,l_{q/2-1}), (i_1,j_1,\ldots, j_{q/2-1}) \in \Gamma_S^{i_1}(i_1, q/2-1)$.
Denote by $r$ the sequence of rows $(l_{q/2-1},\ldots,l_1,i_1,j_1,\ldots,j_{q/2-1})$ then the following estimator is equivalent to the estimator described in Section~\ref{sec:ImpSamplingEstimator}
(slightly abusing the notation therein for concatenating two sequences of rows):
$$Y\eqdef\frac{1}{\tau_{r}} \sum_{\substack{m \in N^{i_1}_S(l_{q/2-1}) \\ \cap N^{i_1}_S(j_{q/2-1})}} c(r,i_q) \sigma(r)\ip{l_{q/2-1}}{m}\ip{j_{q/2-1}}{m}.$$
It is easy to verify that this estimator is unbiased, and that its variance can be bounded using the proof steps of Section \ref{sec:ImpSamplingEstimator}.
This estimator can be implemented algorithmically similarly to
the description in Sections~\ref{sec:Row-orderImportanceSampling} and~\ref{sec:TurnstileImportanceSampling} using less passes over the stream.
Specifically, the above approach decreases the number of ``path'' stages (i.e. all but the ``seed'' sampling stage) by a factor of (roughly) $2$, and the space complexity remains the same up to constant factors. Therefore, we reduce the number of passes over the streams of Algorithm \ref{Alg:row-order} and Algorithm \ref{Alg:turnstile} to $\floor{\frac{p}{4}}+1$ and $\frac{p}{2} + 3$, respectively. This concludes the proofs of Theorems~\ref{thm:row-orderAlgGeneral}, \ref{thm:row-orderAlgSpecial} and \ref{thm:turnstileAlg}.

\section{Pass-Space Trade-off}\label{sec:passSpaceTrade-off}
Very often streaming problems have a sharp transition in space complexity when comparing a single pass to multiple passes over the data. However, it turns out that for the Schatten $p$-norm  of sparse matrices, the space dependence on the number of passes is smooth, allowing one to pick the desired pass-space trade-off. Specifically, for any parameter $s\geq 2$, one can $(1\pm\epsilon)$-approximate the Schatten $p$-norm in $\floor{\frac{p}{2(s+1)}}+1$ passes using ${O}_{p,s}(\epsilon^{-3}k^{2ps}n^{1-\frac{1}{s}})$ words of space.

Recall the Schatten $p$-norm formulation of \eqref{eqn:Schatten-p-Scaled}. This in can be interpreted as partitioning the (contributing) length-$q$ cycles according to their heaviest row, denoted here by $i_1$. Analogously, for any parameter $s \in [2,p-1]$, we split the cycle into $s+1$ segments of hop-distance roughly $\frac{q}{s+1}$, and further partition the cycles according to the heaviest row in each such segment. The idea is to ``cover" a cycle by sampling $s$ rows, where each sampled row is the heaviest among its segment. More precisely, each sample ``covers'' its segment, except for the heaviest row in the entire cycle that will ``cover'' two segments. Then, to evaluate the entire cycle we need $\floor{\frac{q}{s+1}}+1$ passes. The total space needed by the algorithm is ${O}_{p,s}(\epsilon^{-3}k^{2ps}n^{1-1/s})$ words of space, mostly as it computes multiple copies of the estimator (to reduce the variance), similarly to Section \ref{sec:Implement}.

In the first subsection we focus on the case $s=2$ and present a BFS-based algorithm, followed by a brief explanation how to improve the dependence on $k$ by replacing the BFS with adaptive sampling as in the previous sections. In the second subsection we generalize the result to any $s\geq 2$.

\subsection{The Basic Case $s=2$ ($\floor{\frac{p}{6}}+1$ Passes)}
As mentioned, \eqref{eqn:Schatten-p-Scaled} can be interpreted as considering only cycles that ``start'' from the heaviest row of the cycle (by ``rotating'' the cycle). We suggest a variation on this idea. Given a $q$-cycle ``starting'' at the heaviest row $i$, we identify the row $j$ that is the heaviest among the rows at least $q/3$ cycle-hops away from $i$. In other words, if the cycle is $(i=i_1,\ldots,i_q)$, then $j$ is the heaviest among (roughly) $i_{q/3},\ldots,i_{2q/3}$. Therefore, our aim is to sample rows $i$ and $j$ and then to connect four paths: two starting from $i$ and two starting from $j$, each of hop-distance at most $q/3$. As we don't know in advance the hop-distance to row $j$, we store all possible options and only later decide which paths to stich together into a cycle.

Formally, we augment the notation of paths presented in Section \ref{sec:ImpSamplingEstimatorPrelim}. For indices $i,j,i_1 \in [n]$ and integers $t'\leq t''\leq t$, define
$$\Gamma_S^{(i,j;t',t'')}(i_1,t) \eqdef$$
$$\left\{(i_1,\ldots, i_q):\ (i_1,\ldots,i_{t'}) \in \Gamma_S^i(i_1,t'), (i_{t'},\ldots,i_{t''}) \in \Gamma_S^j(i_{t'},t''-t'+1), (i_{t''},\ldots,i_t) \in \Gamma_S^i(i_t'',t-t''+1)\right\}.$$
As we are actually interested in the special case where $t'=\floor{\frac{q}{3}}+1$ and $t''=q-\floor{\frac{q}{3}}$, we shall omit $t',t''$ from the superscript in this special case.

Recall that we focus on cycles in which $i_1=i$, i.e. the heaviest row is the starting of the cycle. Furthermore, we want $j=i_l$ for some $l\in \{\floor{\frac{q}{3}}+2,\ldots,q-\floor{\frac{q}{3}}\}$, i.e. $j$ is part of the cycle, and is at least $\floor{\frac{q}{3}}$ cycle-hops away from $i$.
Accordingly, we can rewrite the Schatten $p$-norm as
\begin{equation}\label{eqn:Schatten-p-twoSet}
\spn{A}{p}=\sum_{i,j}\sum_{\floor{\frac{q}{3}}+2 \leq l \leq q-\floor{\frac{q}{3}}}\sum_{\substack{(i, i_2,\ldots, i_{q}) \\ \in \Gamma_S^{(i,j)}(i):\ i_l=j}} c(i, i_2, \ldots, i_q)\sigma(i, i_2, \ldots, i_q, i).
\end{equation}

We are now ready to present our estimator and an algorithm implementing it. In the algorithm, instead of summing over all $i,j\in [n]$, we sample two multisets $I,J$ and do a BFS of depth $\floor{q/3}$ from each $i\in I$ and $j\in J$, and eventually enumerate over all cycles involving these $i,j$ as in \eqref{eqn:Schatten-p-twoSet}. 
\begin{algorithm}[H]
\caption{Two-Set based Algorithm for Schatten $p$-Norm of $k$-Sparse Matrices for $p \in 2\Z_{\geq 2}$ in Row-Order Stream}\label{Alg:passSpace}

\begin{algorithmic}[1]
\Statex \textbf{Input}: $A \in \R^{n \times n}$ streamed in row-order, $p \in 2\Z_{\geq 2}$, $\epsilon > 0$.
\Statex
\State $r\gets O(\epsilon^{-3}q^{5/2}k^{3p-6}\sqrt{n})$, $Y \gets 0$.
\Parallel{$2r$}
		\Pass{$1$}
			\State sample a row $i \in [n]$ with probability $\tau_i=\frac{\|a_{i}\|_2^q}{\sum_{m} \|a_m\|_2^q}$ \Comment{Using Reservoir Sampling} \label{step:sampleFirstRow}
		\EndPass
		\Pass{$2 \leq t \leq \floor{q/3}+1$}	
			\State store all rows of hop-distance at most $t-1$ from $i$ that have $l_2$-norm smaller than row $i$
		\EndPass
\EndParallel
\State let multisets $I$ and $J$ contain the first and last $r$ samples (from line \ref{step:sampleFirstRow}), respectively
\ForAll{$(i,j)\in I\times J$ such that $\left(\frac{\epsilon}{qk^{2\ceil{q/2}}}\right)^{3/p} \|a_i\|_2\leq \|a_j\|_2 \leq \|a_i\|_2$} \label{step: update}
$$Y\pluseq \frac{1}{\tau_i \cdot \tau_j}\sum_{\floor{\frac{q}{3}}+2 \leq l \leq q-\floor{\frac{q}{3}}}\sum_{\substack{(i, i_2,\ldots, i_{q}) \\ \in \Gamma_S^{(i,j)}(i):\ i_l=j}} c(i, i_2, \ldots, i_q)\sigma(i, i_2, \ldots, i_q, i)$$
\EndFor
\State \textbf{return} $\bar{Y}=\frac{1}{r^2}Y$
\end{algorithmic}
\end{algorithm}

\begin{theorem} \label{thm:TwoSetSchattenAlg}
There exists an algorithm that, given $p \in 2\Z_{\geq 2}$, $\epsilon > 0$ and a $k$-sparse matrix $A \in \R^{n \times n}$ that is streamed in row-order, makes $\floor{\frac{p}{6}}+1$ passes over the stream using at most $O_{p}(\epsilon^{-3}k^{4p}\sqrt{n})$ words of space, and then outputs $\bar{Y}(A)$ that $(1\pm 2\epsilon)$-approximates $\spn{A}{p}$ with probability at least $2/3$.
\end{theorem}

Before the proof, we state the following theorem, which can be viewed as a variant of the Importance Sampling lemma (Theorem \ref{thm:ImportanceSampling}).
\begin{lemma}[Two-Set Sampling]
\label{thm:setSampling}
Let $z=\sum_{i,j\in[n]} z_{i,j} > 0$ for $n\ge 1$,
and suppose the matrix defined by $\{z_{i,j}\}$ is $\Delta$-sparse.%
\footnote{$\Delta$ can be viewed as an upper bound on the in-degrees
  and out-degrees of the directed graph defined by edge weights $z_{ij}$
  on vertex set $[n]$.
}
Let $I,J \in [n]$ be two random multisets of size $r$,
where each of their $2r$ elements is chosen independently with replacement according to the distribution $(\tau_l: l\in [n])$.
Consider the estimator
$$
  Y
  = \frac{1}{r^2} \sum_{i\in I, j\in J} \frac{z_{i,j}}{\tau_{i}\cdot\tau_{j}}.
$$
If $\lambda>0$ satisfies that for all $i,j\in[n]$
both $\tau_i,\tau_j \geq \frac{1}{\lambda}\sqrt{\frac{\abs{z_{i,j}}}{z}}$, then
\begin{equation}\label{eqn:twoSetExpVar}
  \Exp{}{Y} = z
  \text{\ \ and \ \ }
  \Var(Y) \leq \left(\frac{\lambda^2}{r^2}
    + \frac{2\lambda\Delta}{r}\right) z \sum_{i,j\in [n]} \abs{z_{i,j}}.
\end{equation}

\end{lemma}

The proof of Lemma \ref{thm:setSampling} is given in Appendix \ref{App:TwoSetProof}. We now proceed to the proof of Theorem \ref{thm:TwoSetSchattenAlg}, remarking that $k^{O(p)}$ factor can be improved by using the Projection Lemmas, but for simplicity we use more straightforward arguments.
\begin{proof}[Proof of Theorem \ref{thm:TwoSetSchattenAlg}]
First we remark that indeed in $\floor{q/3}+1$ passes all the needed rows of a cycle are kept. For any cycle, row $i$ needs to ``cover" $\floor{q/3}+1+(q-(q-\floor{q/3}))=2\floor{q/3}+1$ rows (including itself), which indeed happens as we do a BFS of size $\floor{q/3}$. Row $j$ must cover at most $q-\floor{q/3}-(\floor{q/3}+2)=q-2\floor{q/3}-2$ rows, including itself. As $\floor{q/3}+1\geq q-2\floor{q/3}-2$, we indeed again cover all possibly needed rows in the $\floor{q/3}+1$ passes. We now go on to prove the approximation bounds.
Let $\beta \eqdef \left(\frac{\epsilon}{qk^{p-2}}\right)^{3/p}$ and define for all $i,j\in [n]$
$$z_{i,j} \eqdef
\begin{cases}
\sum_{\floor{\frac{q}{3}}+2 \leq l \leq q-\floor{\frac{q}{3}}}\sum_{\substack{(i, i_2,\ldots, i_{q}) \\ \in \Gamma_S^{(i,j)}(i):\ i_l=j}} c(i, i_2, \ldots, i_q)\sigma(i, i_2, \ldots, i_q, i) & \quad \text{if } \|a_j\|_2\leq \|a_i\|_2;\\
$0$ & \quad \text{otherwise.}
\end{cases}$$
Then, by Equation \eqref{eqn:Schatten-p-twoSet}, $z' \eqdef \sum_{i,j}z_{i,j}=\spn{A}{p}$. Since line \ref{step: update} in the algorithm considers only pairs $(i,j)$ where $\frac{\|a_j\|_2}{\|a_i\|_2} \in [\beta,1]$, we further define
$$z \eqdef \sum_{i,j:\ \frac{\|a_j\|_2}{\|a_i\|_2} \in [\beta,1] } z_{i,j}.$$

Let us show that the omitted terms do not contribute much to $z'=\spn{A}{p}$, and thus the error introduced by omitting them is small. For simplicity assume $q/3\in \N$, then
\begin{align*}
\left|z'-z\right| &\leq \sum_i \sum_{j: \ \frac{\|a_j\|_2}{\|a_i\|_2}\leq \beta}\left|z_{i,j}\right| \\
&\leq \sum_i \sum_{j: \ \frac{\|a_j\|_2}{\|a_i\|_2}\leq \beta}\sum_{\floor{\frac{q}{3}}+2 \leq l \leq q-\floor{\frac{q}{3}}}\sum_{\substack{(i, i_2,\ldots, i_{q}) \\ \in \Gamma_S^{(i,j)}(i):\ i_l=j}} c(i, i_2, \ldots, i_q)\left|\sigma(i, i_2, \ldots, i_q, i)\right|
\intertext{As $c(i, i_2, \ldots, i_q)\leq q$, and using the conditions on $i$ and $j$ we get}
&\leq q\sum_i \sum_{j: \ \frac{\|a_j\|_2}{\|a_i\|_2}\leq \beta}\sum_{\floor{\frac{q}{3}}+2 \leq l \leq q-\floor{\frac{q}{3}}}\sum_{\substack{(i, i_2,\ldots, i_{q}) \\ \in \Gamma_S^{(i,j)}(i):\ i_l=j}} \|a_i\|_2^{2p/3}\|a_j\|_2^{p/3} \\
\intertext{As each row has at most $k^2$ ``neighboring" rows,}
&\leq k^{2(q-1)}q \beta^{p/3}\sum_i \|a_i\|_2^p = \epsilon \sum_i \|a_i\|_2^p.
\end{align*}
Therefore, using Fact \ref{fact:lengthLeqSchatten-p}, we conclude
\begin{equation}\label{eq:zVsz'}
\abs{z-\spn{A}{p}}\leq \epsilon\spn{A}{p}.
\end{equation}

We proceed to show that the standard deviation of our estimator is bounded by $\epsilon z$, meaning that w.h.p $\bar{Y}\in (1\pm \epsilon)z$, and together with \eqref{eq:zVsz'} this yields $\bar{Y} \in (1\pm 2\epsilon) \spn{A}{p}$. To this end, we want to use Lemma \ref{thm:setSampling} and thus wish to show that
\begin{equation}\label{eq:abs_zij}
\sum_{i,j}\abs{z_{i,j}} \leq 2qk^{2\ceil{q/2}}z
\end{equation}
and that $\lambda \eqdef \sqrt{2qk^{p-4}\frac{n}{\beta^{2p/3}}}=\sqrt{2}q^{3/2}k^{3p/2-4}\frac{\sqrt{n}}{\epsilon}$ satisfies
\begin{equation}\label{eq:var_zij}
\frac{\abs{z_{i,j}}}{z}\leq \lambda^2 \tau_j^2 \qquad \forall i,j\in [n].
\end{equation}
meaning that . We remark that \eqref{eq:var_zij} is indeed sufficient, as $\tau_j\leq \tau_i$, as otherwise $z_{i,j}=0$ and the inequality trivially holds.

To prove \eqref{eq:abs_zij}, we use similar arguments as above, together with \eqref{eq:zVsz'},
\begin{align*}
\sum_{i,j}\abs{z_{i,j}} &\leq q\cdot \sum_{i,j: \ \frac{\|a_j\|_2}{\|a_i\|_2} \in [\beta,1]}\sum_{\floor{\frac{q}{3}}+2 \leq l \leq q-\floor{\frac{q}{3}}}\sum_{\substack{(i, i_2,\ldots, i_{q}) \\ \in \Gamma_S^{(i,j)}(i):\ i_l=j}} \|a_i\|_2^p \\
&\leq qk^{p-2}\sum_{i}{\|a_i\|_2^p} \\
&\leq 2qk^{p-2}z.
\end{align*}

To prove \eqref{eq:var_zij}, fix $i,j$ such that $\frac{\|a_j\|}{\|a_i\|} \in [\beta,1]$, then by similar arguments, together with \eqref{eq:zVsz'} and Fact \ref{fact:lengthLeqSchatten-p},
\begin{align*}
\frac{\abs{z_{i,j}}}{z} &\leq \frac{1}{z}\sum_{\floor{\frac{q}{3}}+2 \leq l \leq q-\floor{\frac{q}{3}}}\sum_{\substack{(i, i_2,\ldots, i_{q}) \\ \in \Gamma_S^{(i,j)}(i):\ i_l=j}} c(i, i_2, \ldots, i_{q})|\sigma(i, i_2, \ldots, i_{q},i)| \\
&\leq \frac{1}{z}\sum_{\floor{\frac{q}{3}}+2 \leq l \leq q-\floor{\frac{q}{3}}}\sum_{\substack{(i, i_2,\ldots, i_{q}) \\ \in \Gamma_S^{(i,j)}(i):\ i_l=j}} q\|a_i\|_2^{2p/3}\|a_j\|_2^{p/3}\\
&\leq qk^{p-4}\frac{\|a_j\|_2^{p}}{\beta^{2p/3}z}\\
&\leq 2qk^{p-4}\frac{\|a_j\|_2^{p}}{\beta^{2p/3}\spn{A}{p}} \\
&\leq 2qk^{p-4}\frac{\|a_j\|_2^{p}}{\beta^{2p/3}\sum_{m}\|a_m\|_2^p}
\intertext{using norm properties (basically applying $\|v\|_q\leq n^{1/q-1/p}\|v\|_p$ to the vector $v=(\|a_1\|_2,\ldots,\|a_n\|_2))$, }
&\leq qk^{p-4}\frac{\|a_j\|_2^{p}}{\beta^{2p/3}(\sum_{m}\|a_m\|_2^q)^2/n} \\
&\leq 2qk^{p-4}\frac{n}{\beta^{2p/3}}\cdot \tau_j^2.
\end{align*}

We further note that for $z_{i,j}$ to be non-zero, row $j$ must be at distance at most $\ceil{q/2}$ from row $i$, and thus each row can participate in at most $k^{2\ceil{q/2}}$ different non-zero $z_{i,j}$, i.e., $\Delta\leq k^{p/2-2}$. Combining all the above, we conclude that setting $r=O(\epsilon^{-2}\lambda \Delta)\cdot 2qk^{p-2}=O(\epsilon^{-3} q^{5/2}k^{3p-6}\sqrt{n})$ will give w.h.p a $(1\pm 2\epsilon)$-approximation to the Schatten $p$-norm by Chebyshev's inequality.

As for each row in $I\cup J$ the algorithm stores neighborhoods of size $O\left((k^2)^{q/3}\right)$, and storing each row in the neighborhood takes $O(k)$ words, there is an extra factor of $k^{p/3+1}$. Thus the total space is $O(\epsilon^{-3}q^{5/2}k^{10p/3-5}\sqrt{n})$ words.
\end{proof}

\begin{remark}
As mentioned earlier, the BFS approach can be replaced with the adaptive sampling approach from previous sections. For the first $r$ samples (in $I$), the algorithm adaptively samples two paths of hop-distance (roughly) $q/3$, similarly to Section \ref{sec:reducingPasses}. For each of the last $r$ samples (in $J$), the algorithm chooses $\rho \in [q/3]$ uniformly at random (and independently of all other steps), and adaptively samples a path of hop-distance $\rho$ and a path of hop-distance (roughly) $\frac{q}{3}-\rho$. It then tries to ``stitch" these paths to create $q$-cycles. The bound on $\lambda$ (i.e. \eqref{eq:var_zij}) increases by a factor of $q/3$ due to $\rho$ (this can be viewed as replacing the BFS with multiple random paths), but as the algorithm does not keep the entire neighborhoods, a $k^{p/3}$ factor is shaved from the space complexity. This, together with a tighter analysis, can improve the dependence on $k$ in Theorem \ref{thm:TwoSetSchattenAlg} to $k^{19p/8+O(1)}$.
\end{remark}

\subsection{General $s$ (using $\floor{\frac{p}{2(s+1)}}+1$ Passes)}
We generalize the algorithm from the previous subsection, such that given some integer $s \in [2, p-1]$, the algorithm samples in parallel in the first pass $r\cdot s$ rows for $r=O_{p,\epsilon,s}(k^{4p}n^{1-1/s})$, where each ``seed" row $i$ is sampled with probability $\tau_i=\frac{\|a_i\|_2^{p/s}}{\sum_m \|a_m\|_2^{p/s}}$. In the following passes it runs a BFS of depth (roughly) $\frac{q}{s+1}$, keeping all the shorter rows (in $l_2$-norm) in the neighborhood of each seed. The first $r$ samples are denoted as multiset $I$, and the other samples are split into $s-1$ multisets of size $r$ denoted as $J_1,\ldots, J_{s-1}$. The algorithm then considers $s$-tuples $(i,j_1,\ldots,j_{s-1})$ where $i\in I$ and every row $j_u\in J_u$ has $l_2$-norm in the range $(\beta', 1)$ relative to that of row $i$, for $\beta'\approx\left(\frac{\epsilon}{sqk^p}\right)^{(s+1)/p}$. The estimator is formed by looking at the eligible $s$-tuples, and for each such tuple adding the contributions of all the $q$-cycles obtained by ``stitching" paths of hop-distance (roughly) $\frac{q}{s+1}$ passing through these seeds, as follows:
$$Y\pluseq \frac{1}{\tau_i\tau_{j_1}\cdots\tau_{j_{s-1}}}\sum_{\frac{q}{s+1}\leq l_1 \leq \frac{2q}{s+1}}\cdots \sum_{\frac{(s-1)q}{s+1} \leq l_{s-1} \leq \frac{s\cdot q}{s+1}}\sum_{\substack{(i, i_2,\ldots, i_{q}) \\ \in \Gamma_S^{(i,j_1,\ldots,j_{s-1})}(i): \\ i_{l_1}=j_1,\ldots, i_{l_{s-1}}=j_{s-1}}} c(i, i_2, \ldots, i_q)\sigma(i, i_2, \ldots, i_q, i).$$
The algorithm's final output is $\bar{Y}=\frac{1}{r^s}Y$.

\begin{theorem} \label{thm:sSetSchattenAlg}
There exists an algorithm that, given $p \in 2\Z_{\geq 2}$, $\epsilon > 0$,
an integer $s \in [2, p-1]$ and a $k$-sparse matrix $A \in \R^{n \times n}$
streamed in row-order,
makes $\floor{\frac{p}{2(s+1)}}+1$ passes over the stream
using $O_{p}\left(\epsilon^{-3}k^{2ps}n^{1-\frac{1}{s}}\right)$ words of space,
and outputs $\bar{Y}(A)$ that $(1\pm 2\epsilon)$-approximates $\spn{A}{p}$
with probability at least $2/3$. 
\end{theorem}

\begin{proof}[Proof Sketch]
The proof follows similar steps as the proof for $s=2$. First, the error introduced by taking only certain cycles changes, as now we miss cycles in which at least one of the sampled $j_u$ is smaller than $\beta'$. However their total contribution can be bounded by $(s-1)(\beta') ^{p/(s+1)}qk^{p-2}<\epsilon$ relative to $\spn{A}{p}$. Next, an $s$-Set Sampling Lemma is proved using the same arguments as the Two-Set Sampling Lemma. It asserts that the estimator
$$Y=\frac{1}{r^s}\sum_{i\in I, j_1\in J_1,\ldots, j_{s-1}\in J_{s-1}} \frac{z_{i,j_1,\ldots,j_{s-1}}}{\tau_i \tau_{j_1}\cdots \tau_{j_{s-1}}}$$
is unbiased, and that if $\lambda>0$ satisfies that for every $i,j_1,\ldots,j_{s-1}\in [n]$, all $\tau_i,\tau_{j_1},\ldots,\tau_{j_{s-1}}\geq \frac{1}{\lambda}\left(\frac{\abs{z_{i,j_1,\ldots,j_{s-1}}}}{z}\right)^{1/s}$, then
$$\text{Var}(Y)\leq O\left(\left(\Delta+\frac{\lambda}{r}\right)^s-\Delta^s\right)z\sum_{i,j_1,\ldots,j_{s-1}\in [n]} \abs{z_{i,j_1,\ldots,j_{s-1}}}.$$

The proof for the inequality analogous to \eqref{eq:abs_zij}, which bounds the ratio between the absolute sum of $z_{i,j_1,\ldots,j_{s-1}}$ and $z$, is the same. To prove the bound $\lambda$ (i.e. analogous to \eqref{eq:var_zij}), we need to bound the shortest $j_u$ among rows $(j_1,\ldots,j_{s-1})$. To do so we first bound all ``seeds" except $j_u$ using row $i$, and then use the same arguments that result in $\lambda=\left(C_{\epsilon}qk^p\frac{n^{s-1}}{(\beta')^{2p/(s+1)}}\right)^{1/s}$ for a suitable constant $C$ dependent on $\epsilon$. Finally, now each $i$ can have $(s-1)k^{2\ceil{q/2}}$ different $(j_1,\ldots,j_{s-1})$, i.e. $\Delta \leq (s-1)k^{q+2}$. Picking $r =O\left( \epsilon^{-3} s \Delta^{s-1}\lambda\right)$ results in the desired approximation.

The space complexity analysis is as in the proof of Theorem \ref{thm:TwoSetSchattenAlg}, resulting in 
$$O\left(\epsilon^{-3}(s-1)^s\cdot q^{2+1/s}\cdot k^{p(s/2+11/6+1/s)+2s-O(1)}\cdot n^{1-1/s}\right)$$
words of space. 
\end{proof}

\section{Lower Bound for One-Pass Algorithms in the Row-Order Model}\label{sec:row-orderlowerbounds}

We show a space lower bound of $\Omega(n^{1-4/{\lfloor p \rfloor_4}})$ bits for one-pass algorithms and even $p$ values in the row-order model. Our main technical contribution is the analysis of even $p$ values in a reduction presented in \cite{LW16a}, based on the Boolean Hidden Hypermatching \cite{VY11, BS15}.
Although this is not mentioned in~\cite{LW16a}, it can easily be verified from the proof of~\cite[Theorem $3$]{LW16a} (stated below as Theorem \ref{thm:LW_LB}) that this reduction applies also to the row-order model.\footnote{In fact, also Theorem $4$ in~\cite{LW16a} applies to row-order streams, providing a different proof for the $\Omega(n^{1-\varepsilon})$ lower bound for $p\notin 2\Z$ proved in \cite{BCKLWY18}.} Our bound is closely related to the $\Omega(n^{1-1/\varepsilon})$ bits lower bound  for $p\notin 2\Z$, proved in \cite{BCKLWY18}, and is also nearly tight with the upper bound from the same paper (see discussion at the end of this section).

We first recall the definitions presented in \cite{LW16a}. Let $D_{m,l} \ (\text{for} \ 0\leq l\leq m)$ be an $m\times m$ diagonal matrix with the first $l$ diagonal elements equal to $1$ and the remaining diagonal entries equal to $0$, and let $\textbf{1}_m$ be an $m$-dimensional vector full of $1$s, thus $\textbf{1}_m\textbf{1}_m^\top$ is the $m\times m$ all-ones matrix. Define
$$M_{m,l}(\gamma)=\left(\begin{array}{cc}
\textbf{1}_m\textbf{1}_m^\top & 0\\
\sqrt{\gamma}D_{m,l} & 0
\end{array}\right) ,$$
where $\gamma >0$ is a constant (which may depend on $m$).

Let $m\geq 2$ be an even integer, and let $p_m(l) \eqdef \binom{m}{l}/2^{m-1}$ for $0\leq l\leq m$. Let $\mathcal{E}(m)$ be the probability distribution defined on even integers $\{0,2,\ldots,m\}$ with probability density function $p_m(l)$. Similarly, let $\mathcal{O}(m)$  be the distribution on odd integers $\{1,3,\ldots,m-1\}$ with density function $p_m(l)$. We say that a function $f$ on square matrices is \emph{diagonally block-additive} if $f(X)=f(X_1)+\ldots+f(X_s)$ for any block diagonal matrix $X$ with square blocks $X_1,\ldots,X_s$. As noted in \cite{LW16a}, $f(X)=\spn{X}{p}$ is diagonally block-additive.

We observe that the reduction presented in~\cite{LW16a} is applicable also to row-order streams, and thus state below a slightly stronger version of Theorem 3 from that paper.

\begin{theorem}[Theorem 3 in \cite{LW16a}]\label{thm:LW_LB}
Let $t$ be an even integer and let $f$ be a function of square matrices that is diagonally block-additive. If there exists $m=m(t)$ and $\gamma=\gamma(m) > 0$, such that the following ``gap condition'' holds:
\begin{equation} \label{eq:exp_gap}
\Exp{l\sim\mathcal{E}(t)}{f\left(M_{m,l}(\gamma)\right)}-\Exp{l\sim\mathcal{O}(t)}{f\left(M_{m,l}(\gamma)\right)}\neq 0,
\end{equation}
then there exists a constant $\varepsilon=\varepsilon(t) > 0$ such that every \textbf{row-order} streaming algorithm that, given $X\in \R^{N\times N}$ (for sufficiently large $N$), approximates $f(X)$ within factor $1\pm \varepsilon$ with constant error probability, must use $\Omega_t(N^{1-1/t})$ bits of space.
\end{theorem}

We can now present our analysis for even $p$ values.

\begin{lemma}\label{thm:singlePassLBLemma}
Let $f(X)=\spn{X}{p}$, for $p\in 4\Z_{\geq 1}$. Then the gap condition \eqref{eq:exp_gap} is satisfied, under the choice $m=t$ and $\gamma=1$, if and only if $t\leq p/4$.
\end{lemma}

\begin{proof}
As shown in the proof of Theorem 4 in \cite{LW16a}, for $m=t$ and $\gamma=1$, the non-zero singular values of a block $M_{t,l}(1)$ are as follows. For $l=0$, the only non-zero singular value is $t$. For $0<l<t$, the non-zero singular values are $r_1(l)=\sqrt{\frac{(t^2+1)+\sqrt{(t^2-1)^2+4lt}}{2}}$, $r_2(l)=\sqrt{\frac{(t^2+1)-\sqrt{(t^2-1)^2+4lt}}{2}}$ and $1$ with multiplicity $l-1$. And for $l=t$, the non-zero singular values are $r_1(t)=\sqrt{\frac{(t^2+1)+\sqrt{(t^2-1)^2+4t^2}}{2}}$ and $1$ with multiplicity $t-1$. Further note that that $r_2(t)=0$.
Using this, and recalling the distribution of the blocks, the left-hand side of the gap condition \eqref{eq:exp_gap} is
\begin{equation} \label{eq:eq_for_lower_bound_lemma}
\frac{1}{2^{t-1}}\left[t^p+\sum_{\text{even} \ l}{t\choose l}\left( (l-1)+r_1^{p}(l)+r_2^{p}(l)\right) -\sum_{\text{odd} \ l}{t\choose l}\left((l-1)+r_1^{p}(l)+r_2^{p}(l)\right)\right]
\end{equation}
and we can rewrite this as
$$\frac{1}{2^{t-1}}\left[t^p+\sum_{0<l\leq t}{t\choose l}(-1)^l(l-1)+\sum_{0<l\leq t}{t\choose l}(-1)^l\left(r_1^{p}(l)+r_2^{p}(l)\right)\right].$$
For the first sum, by Corollary $2$ in \cite{R96}, we know that
$$\sum_{l=0}^t (-1)^l\binom{t}{l}(l-1)=0$$
meaning that
$$\sum_{0<l\leq t}{t\choose l}(-1)^l(l-1)=1.$$
Let $q=p/2$. It holds that
\begin{equation*}
\begin{split}
 r_1^p(l)+r_2^p(l) & = \left(\frac{(t^2+1)+\sqrt{(t^2-1)^2+4lt}}{2}\right)^{q}+\left(\frac{(t^2+1)-\sqrt{(t^2-1)^2+4lt}}{2}\right)^q\\
\intertext{and using the binomial theorem,}
& = \frac{1}{2^q}\left[\sum_{i=0}^q \left(t^2+1\right)^{q-i}\left(\sqrt{(t^2-1)^2+4lt}\right)^i+ \sum_{i=0}^q \left(-1\right)^i \left(t^2+1\right)^{q-i}\left(\sqrt{(t^2-1)^2+4lt}\right)^i\right]. \\
\intertext{We note that the alternating sum double the even values the zero out the odd values, thus the above can be rewritten as}
& = \frac{1}{2^{q-1}}\sum_{\text{even }i} {q\choose i} \left(t^2+1\right)^i\left((t^2-1)^2-4lt\right)^{\frac{q-i}{2}}. \\
\intertext{and by applying it again, on the second multiplicative term,}
 & =  \frac{1}{2^{q-1}}\sum_{\text{even }i} {q\choose i} \left(t^2+1\right)^i\sum_{j=0}^{\frac{q-i}{2}}{\binom{\frac{q-i}{2}}{j}}\left(t^2-1\right)^{2j}\cdot\left(4t\right)^{\frac{q-i}{2}-j}\cdot l^{\frac{q-i}{2}-j}.
\end{split}
\end{equation*}
Combining the two insights results in
$$\eqref{eq:eq_for_lower_bound_lemma} =\frac{1}{2^{t-1}}\left[t^p+1+ \sum_{l=1}^t (-1)^l \left(\frac{1}{2^{q-1}}\sum_{\text{even }i} {q\choose i} (t^2+1)^i\sum_{j=0}^{\frac{q-i}{2}}{\frac{q-i}{2}\choose j}(t^2-1)^{2j}\left(4tl\right)^{\frac{q-i}{2}-j}l^{\frac{q-i}{2}-j}\right)\right].$$

We further note that for $l=0$, the term in the inner parentheses is non-zero only when $\frac{q-i}{2}= j$. In this case we get, using the binomial theorem once more,
$$\frac{1}{2^{q-1}}\sum_{\text{even }i} {q\choose i} (t^2+1)^i(t^2-1)^{q-i}=\left(\frac{t^2+1+t^2-1}{2}\right)^q+\left(\frac{t^2+1-t^2+1}{2}\right)^q=1+t^p.$$

Therefore, we can rewrite \eqref{eq:eq_for_lower_bound_lemma} as
\begin{equation*}
\begin{split}
\eqref{eq:eq_for_lower_bound_lemma}  &=\frac{1}{2^{t-1}}\left(\sum_{l=0}^t (-1)^l \frac{1}{2^{q-1}}\sum_{\text{even }i} {q\choose i} (t+1)^i\sum_{j=0}^{\frac{q-i}{2}}{\frac{q-i}{2}\choose j}(t-1)^{2j}4^{\frac{q-i}{2}-j}l^{\frac{q-i}{2}-j}\right)\\
\intertext{and using \cite{LW16a} observation,}
&=\frac{1}{2^{t-1}}(-1)^t t! \sum_{\text{even } i } {q\choose i} (t+1)^i\sum_{j=0}^{\frac{q-i}{2}}{\frac{q-i}{2}\choose j}(t-1)^{2j}4^{\frac{q-i}{2}-j}\stirling{\frac{q-i}{2}}{t}
\end{split}
\end{equation*}
where $\stirling{\frac{q-i}{2}}{t}$ are Stirling numbers of the second kind. As for a fixed $t$ all terms are of the same sign, the sum vanishes only when $\stirling{\frac{q-i}{2}}{t}=0 \ \forall i$, which happens when $t> q/2=p/4$.
 \end{proof}

We remark that Lemma \ref{thm:singlePassLBLemma} extends to $p\equiv2\pmod4$ when $t\leq(p-2)/4$, by replacing in the proof $q=p/2$ with $\tilde{q}=(p-2)/2$. The next theorem follows easily by combining Theorem \ref{thm:LW_LB} and Lemma \ref{thm:singlePassLBLemma}.

\singlePassLBThm*

\begin{proof}
Let us first assume that $p\equiv0\pmod4$. As shown in Lemma \ref{thm:singlePassLBLemma}, the gap condition \eqref{eq:exp_gap} holds for $f(X)=\spn{X}{p}$ and $t=p/4$, thus by Theorem \ref{thm:LW_LB} the space complexity is $\Omega(n^{1-1/t})=\Omega(n^{1-4/p})$ bits. For $p\equiv2\pmod4$ the  claim holds for $t=(p-2)/4$, yielding an $\Omega(n^{1-4/(p-2)})$ bits lower bound.
\end{proof}
We note that for $p\equiv0\pmod4$ the above matches up to logarithmic factors the upper bound for the row-order algorithm presented in \cite{BCKLWY18}, i.e. tight for matrices in which every row and column has $O(1)$ non-zero elements. For $p\equiv2\pmod4$, there is a small gap: the lower bound is $\Omega(n^{1-4/(p-2)})$ while the upper bound obtained in \cite{BCKLWY18} is $\tilde {O}_k (n^{1-4/(p+2)})$.

\section{$O_\epsilon(1)$-Space Algorithm for Schatten $4$-Norm of General Matrices}\label{sec:row-orderSchatten4}

We present an $O(1/\epsilon^2)$-space algorithm for  $(1+\epsilon)$-approximation of the Schatten $4$-norm in the row-order model. As this result does not depend on the sparsity and is applicable to any matrix, it significantly improves the previously known row-order algorithm, presented in \cite{BCKLWY18} that uses space $\tilde{O}_{p,\varepsilon}(k)$, and is also better than the result of Section \ref{sec:Row-orderImportanceSampling}.

The algorithm exploits the fact that $A^\top A=\sum_i a_i^\top a_i$ (i.e. summing over the outer product of every row with itself), to sketch the Frobenius norm $\sum_{j_1,j_2}((A^\top A)_{j_1,j_2})^2=\|A^\top A\|_F^2=\spn{A}{4}$. To do so, it uses two random $4$-wise independent vectors, following an idea presented in \cite{IM08} (extending the classic \cite{AMS99} result), as follows.
\begin{lemma}[Lemma 3.1 in \cite{IM08}] \label{lemma: IM_result}
Consider random $h,g\in\{-1,1\}^n$ where each vector is $4$-wise independent (and independent of the other one). Let $v\in \R^{n^2}$ and $z_j=h_{j_1}g_{j_2}$ for $j\in[n]^2$, and define $\Upsilon=(\sum_{j\in [n]^2} z_j v_j)^2$. Then
$$\Exp{}{\Upsilon}=\sum_{j\in [n]^2} v_j^2 \text{,\ \ and \ \ } \text{Var}(\Upsilon)\leq 3(\Exp{}{\Upsilon})^2.$$
\end{lemma}

\begin{algorithm}[H]
\caption{Algorithm for Schatten $4$-Norm of General Matrices in Row-Order Streams}\label{Alg: Schatten4}
\begin{algorithmic}[1]
\Statex \textbf{Input}: $A \in \R^{n \times n}$ streamed in row-order, $\epsilon > 0$.
\Statex
\Parallel{$m=\tilde{O}(1/\epsilon^2)$}	
	\State init: $Y \leftarrow 0$ and choose two random $4$-wise independent vectors $h,g\in \{-1,1\}^n$
	\State upon receiving row $a_i$, update: $Y \pluseq \langle h,a_i \rangle \langle g,a_i \rangle$
	\State let $\Upsilon \leftarrow Y^2$
\EndParallel
\State \textbf{return} average of the $m$ copies of $\Upsilon$
\end{algorithmic}
\end{algorithm}

\begin{theorem}\label{Thm: Schatten4analysis}
Suppose that $A\in \R^{n\times n}$ is a matrix given in one-pass row-order model. Algorithm \ref{Alg: Schatten4} uses space $O(1/\epsilon^2)$ and outputs a $(1+\epsilon)$-approximation to $\spn{A}{4}$ with probability at least $2/3$.
\end{theorem}

\begin{proof}
Consider one copy of the independent sketches. Using simple manipulations, we can write:
$$Y=\sum_i\left(\sum_{j_1}h_{j_1}A_{i,j_1}\right)\left(\sum_{j_2}g_{j_2}A_{i,j_2}\right)=\sum_{j_1,j_2}h_{j_1} g_{j_2} (A^\top A)_{j_1,j_2}$$
By looking at $A^\top A$ as vector of dimension $n^2$, it easily follows from \ref{lemma: IM_result} that $\Exp{}{\Upsilon}=\|A^\top A\|_F^2=\spn{A}{4}$ and $\text{Var}(\Upsilon)\leq 3\|A\|_{S_4}^8$. Repeating the sketch $O(1/\epsilon^2)$ times independently, decreases the variance and gives the desired result (by Chebyshev's inequality).
\end{proof}

\section{Applications}\label{sec:Applications}

In this section we present two applications of our Schatten-norm algorithms 
to some common functions of the spectrum, 
by approximating these functions using low-degree polynomials. 
We employ the well-known idea that
just as functions $f:\R\to\R$ can be approximated 
in a specific interval by polynomials arising from a Taylor expansion
(or using Chebyshev polynomials), 
spectral functions can be approximated by matrix polynomials 
if the matrix eigenvalues lie in a bounded range.
We just need to implement this method in a space-efficient streaming fashion. 
In some applications we require the input matrix to have a bounded spectrum.
Unfortunately, there is no known streaming algorithm to estimate the spectrum
of an input matrix.

\subsection{Approximating Spectral Sums of Positive Definite Matrices}\label{sec:spectralSums}

We demonstrate how our Schatten-norm estimators can be used to approximate commonly used spectral functions of sparse matrices presented as a data stream.
We consider three different spectral functions,
log-determinant, trace of matrix inverse and Estrada index of a Laplacian matrix,
that all belong to the class of spectral sums, as defined below.
These results apply to sparse matrices that are either
positive definite (PD), positive semidefinite (PSD), or Laplacian. 
Throughout, $\log x$ denotes the natural logarithm of $x$.

\begin{definition}[Spectral Sums~\cite{HMAS17}]
Given a function $f: \R\rightarrow \R$ and a matrix $A\in \R^{n\times n}$ with real eigenvalues $\lambda_1,\ldots,\lambda_n$, a spectral sum is defined as
$$S_f(A)=\tr(f(A)) = \sum_{i=1}^{n} f(\lambda_i).$$

When $f(x)=\log x$, the sum is known as log-determinant,
when $f(x)=1/x$ it is known as the trace of the matrix inverse,
and when $f(x)=\exp(x)$ it is known as Estrada index.
\end{definition}

\begin{theorem}\label{Thm:spectralSums}
For every spectral function $S_f$ from the table below,
there is an algorithm with the following guarantee.
Given as input $\epsilon,\theta>0$, 
and a $k$-sparse matrix $A\in \R^{n\times n}$ presented as a row-order stream
and whose eigenvalues all lie in the interval $I_f$ given in the table, 
the algorithm makes $\floor{m_f/4}+1$ passes over the stream using $W_f$ words of space and outputs an estimate $\rho(A)$ such that
$$ \Pr\big[ \rho(A)\in (1\pm 2\epsilon)S_f(A) \big] \ge 2/3.$$
\begin{center}
\begin{tabular}{ |c|c|c|c|}
\hline
\boldmath $S_f$ & \boldmath $I_f$ & \boldmath $m_f$ & \boldmath $W_f$\\
Spectral Function & Spectrum Interval & & Words of Space\\
\hline\hline

\normalfont{Log-Determinant}  & $[\theta,2)$ & $\ceil{\frac{1}{\theta}\cdot \log \frac{1}{\epsilon}}$ & $O_{m_f}(\epsilon^{-2}k^{3m_f/2 - 3})$\\
\hline
\normalfont{Trace of Matrix Inverse}  & $[\theta,2)$ & $\ceil{\frac{1}{\theta}\cdot \log \frac{1}{\epsilon}}$ & $O_{m_f}(\epsilon^{-2}k^{3m_f/2 - 3})$\\
\hline
\normalfont{Estrada Index of a Laplacian} & $[0,\theta]$ \footnote{If the underlying graph is unweighted then the largest eigenvalue is bounded by the degree, i.e. $\theta \leq 2k$.} & $\ceil{(e\theta+1)\log \frac{1}{\varepsilon}-1}$ & $O_{m_f}(\epsilon^{-2}k^{m_f/2})$\\
\hline
\end{tabular}
\end{center}
\end{theorem}

At a high level, the proof follows that of Boutsidis et al.~\cite{BDKKZ17}, 
who present a time-efficient algorithm for approximating the log-determinant
of PD matrices.
Besides extending their method to two other spectral sums,
the main difference is that we need to adapt their argument
to be space-efficient in the streaming model.
More specifically, the log-determinant of a PD matrix is approximated
in~\cite[Lemma 7]{BDKKZ17} 
by a truncated version (i.e., only the first summands) of its Taylor expansion
\begin{equation}\label{eq:TaylorLogDet}
  \log\det(A)=-\sum_{p=1}^{\infty}\tr((I_n-A)^p)/p.
\end{equation}
They then approximate the required matrix traces by multiplying the respective
matrix by a Gaussian vector (from left and right),
which can be implemented faster than matrix powering, 
by starting with the vector and repeatedly multiply it by a matrix. 
While this is time-efficient, it is not space-efficient when the input matrix is sparse, in which case our streaming algorithm has better space complexity.
One other difference to note is that our algorithm approximates
each of the above-mentioned traces \emph{separately},
and thus we need all the Taylor expansion coefficients to be non-negative,
which indeed applies for these three spectral functions.%
\footnote{Our method extends to alternating Taylor sums if the coefficients decrease by a constant factor, by bounding the approximation error difference of every two consecutive summands. One such an example is $\tr(\exp({-A}))$.
}

\begin{proof} 
We follow the proof framework of Lemma $8$ in \cite{BDKKZ17},
achieving the desired relative error of the desired spectral function using a truncated version of its Taylor expansion, consisting $m_f$ terms.
The first relative error is due to the tail of the series,
i.e. the terms that were not considered in the final estimate. 
For the log-determinant, the above-mentioned Lemma $8$ shows that it suffices to $(1 \pm \epsilon)$-approximate the first 
$m_f=\lceil \frac{1}{\theta}\cdot \log \frac{1}{\epsilon}\rceil$ 
terms of its Taylor expansion \eqref{eq:TaylorLogDet}
in order to obtain a $(1 \pm 2\epsilon)$-approximation of $\log \det(A)$. 
The same proof holds also for the Taylor expansion
$$\tr(A^{-1})=\sum_{p = 1}^{\infty} -\tr((I_n-A)^p)$$
and obtaining a $(1 \pm 2\epsilon)$-approximation of $\tr(A^{-1})$
(for the same value of $m_f$).
To achieve this error bound for the Estrada index of a Laplacian, 
the number of values of its Taylor series (see e.g.~\cite{DL11, GDR11})
\begin{equation}\label{eq:EstradaTaylor}
\tr(\exp(A))=\sum_{p=0}^{\infty}\tr(A^p)/{p!}
\end{equation}
that need to be approximated
is $m_f= \ceil{(e\theta+1)\log \frac{1}{\varepsilon}-1}$,
as shown in Appendix \ref{App:LaplacianEstradaIndexProof}.

We are left to explain how to $(1\pm \varepsilon)$-approximate the first $m_f$ values of the Taylor expansion (causing the other relative error). Recall that if a matrix $B$ is PSD then $\tr(B^p)=\sum \lambda_i^p=\spn{B}{p}$, where $\lambda_1,\ldots,\lambda_n\geq0$ are its eigenvalues. Furthermore, for such matrices our algorithm works for every integer $p\geq 2$, and therefore this relative error is immediate from Theorems \ref{thm:row-orderAlgGeneral} and \ref{thm:row-orderAlgSpecial}.

To conclude, we can compute the traces of $B^p$ in parallel for all integer $p=2,3,\ldots,m_f$ using Algorithm \ref{Alg:row-order}, while for $p=1$ one can compute $\spn{B}{1}=\tr(B)$ by directly summing the main diagonal entries.
These parallel executions take $\floor{m/4}+1$ passes and the total space is at most $O_{m}(\epsilon^{-2}k^{3m/2 - 3})$ words of space for log-determinant and trace of matrix inverse, and $O_{m}(\epsilon^{-2}k^{m/2})$ words of space for the Estrada index of a Laplacian matrix.
\end{proof}

\subsection{Approximating the Spectrum of PSD matrices}
We present an application of our algorithm to (weakly) estimate the spectrum of a matrix, with eigenvalues bounded in $[0, 1]$ using approximations of a ``few'' Schatten norms of the matrix. This is based on the work of Cohen-Steiner et al.~\cite{SKSV18} on approximating the spectrum of a graph which is in turn based on insightful work by Wong and Valiant \cite{WV16} on approximately recovering a distribution from its moments using the Moment Inverse method. 

Fix a PSD matrix $A \in \R^{n \times n}$ with eigenvalues $1 \geq \lambda_1 \geq \ldots \geq \lambda_n$ and define the $l$-th moment of the spectrum to be $\frac{1}{n}\spn{A}{l} = \frac{1}{n}\sum_{i \in [n]} \lambda_i^l$. Cohen-Steiner et al.\ show that estimating $O(1/\epsilon)$ moments of $A$ up to multiplicative error $O(\epsilon)$ is sufficient to estimate the spectrum of $A$ within earth-mover distance $O(\epsilon)$. It is well-known that the the $L_1$ distance between two sorted vectors of length $n$ is exactly $n$ times the earth-mover distance between the corresponding point-mass distributions (uniform probability on each of the $n$ indices). Hence, the recovery scheme of Cohen-Steiner et al.\ allows us to recover the spectrum within $L_1$ distance $O(\epsilon n)$ by estimating only $O\left( \frac{1}{\epsilon} \right)$ moments of the matrix $A$. Specifically, we get the following result,  

\begin{theorem}[Theorem 7 in \cite{SKSV18}]\label{thm:spectrumFromMoments}
Given a constant $\epsilon > 0$, there exists a parameter $s = \frac{C}{\epsilon}$ (where $C > 0$ is an absolute constant) and an algorithm $R$ such that, for a PSD matrix $A \in \R^{n \times n}$ with eigenvalues $\lambda = (\lambda_1, \dots, \lambda_n) \in [0,1]^n$ and a vector $y \in \R^s$ with the property that $y_i = \|\lambda\|_i^i \pm \exp(-C'\epsilon)$ for all $i \in [s]$ and absolute constant $C' > 0$, $R$ reads $y$ and outputs a vector $\hat{\lambda}$ such that $\|\lambda  - \hat{\lambda}\|_1 \leq \epsilon n$. 
\end{theorem}

For an error parameter $\epsilon > 0$ and parameter $s = \frac{C}{\epsilon}$ (where $C > 0$ is an absolute constant) as defined in the above theorem, given a $k$-sparse PSD matrix $A \in \R^{n \times n}$ that is streamed in row-order and whose eigenvalues are in the range $[0,1]$, one can use Algorithm \ref{Alg:row-order} to compute the vector $y \in \R^s$ with the desired guarantee using space $O(k^{3s/2 - 3}\exp(-C'\epsilon))$ for some absolute constant $C' > 0$ and using $\floor{s/4}+1$ passes over the stream.

\section{Experiments}\label{sec:experiments}
In this section we present numerical experiments illustrating the performance of the row-order Schatten $p$-norm estimator described in Section \ref{sec:Row-orderImportanceSampling}. We simulate the row-order stream by reading the input matrix row by row. The results not only follow theoretical space bounds, showing that the algorithm is indeed independent of the matrix size, but are actually several orders of magnitude better. In addition, the experiments show that the algorithm is robust to noise, and these two results suggest that real-life behavior of the algorithm is significantly better than our theoretical bounds.

The inputs used are $\{0,1\}^{n\times n}$ matrices, representing collaboration network graphs (nodes represent scientists and edges represent co-authoring a paper) from the e-print arXiv for scientific collaborations in five different areas in Physics. The data was obtained from the Stanford Large Network Dataset Collection \cite{StanData} which was in-turn obtained from \cite{10.1145/1217299.1217301}. In order to study the effect of sparsity, we ``sparsify'' each (of five) matrix by sampling $10$ nonzero entries in each row uniformly at random (note that max column-sparsity can be larger than $10$).

In the first experiment, we use the arXiv General Relativity and Quantum Cosmology collaboration network which has $n=5242$ rows and columns; after ``sparsifying'' the matrix as mentioned, the max column-sparsity is $37$ and the average sparsity is $6.1$. We fix the value of $p$ to be $6$, and using our algorithm from Section \ref{sec:Row-orderImportanceSampling}, we vary number of estimators (walks) $t$ and compute the \emph{relative error} of the average of the $t$ walks. We repeat this process $10$ times for every value of $t$ and plot the mean and standard deviation in Figure \ref{fig:exp1}. 
In addition, we show in this figure the results of running the same experiment on a ``noisy'' version of the matrix, by adding to it an error matrix where $1/5$ of the entries are drawn independently from~$\mathcal{N}(0, 0.1^2)$\footnote{This value assures the $l_2$-norm of the error in a row is ``comparable'' to the $l_2$-norm of the data: $(0.1)^2 \times 5242 \times 0.2 \approx 10 =\text{max row-sparsity}$.}.

\begin{figure}[!htb]
  \centering
  \begin{minipage}[t]{0.49\textwidth}
    \includegraphics[width=\linewidth]{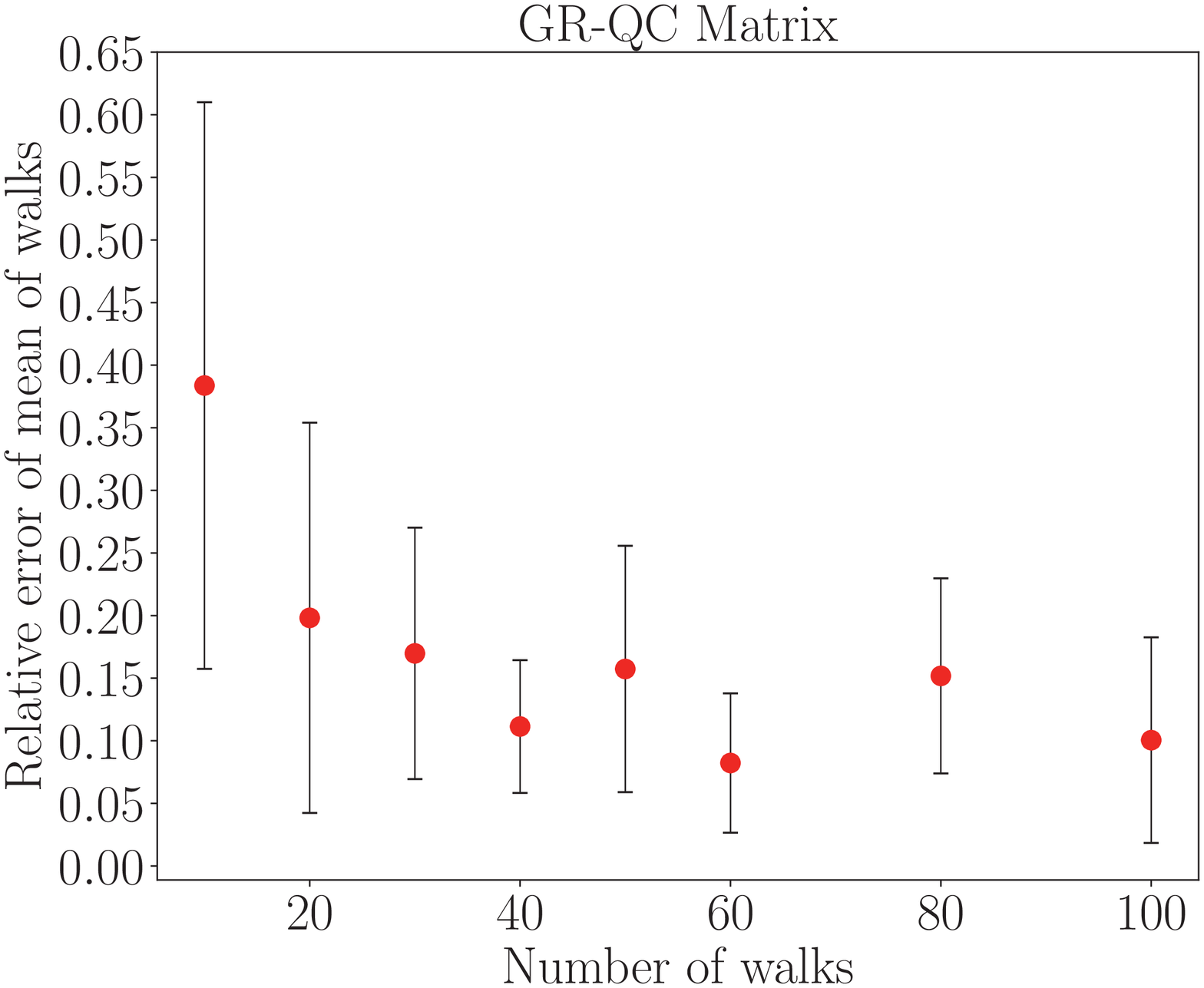}
  \end{minipage} 
  \begin{minipage}[t]{0.49\textwidth}
    \includegraphics[width=\linewidth]{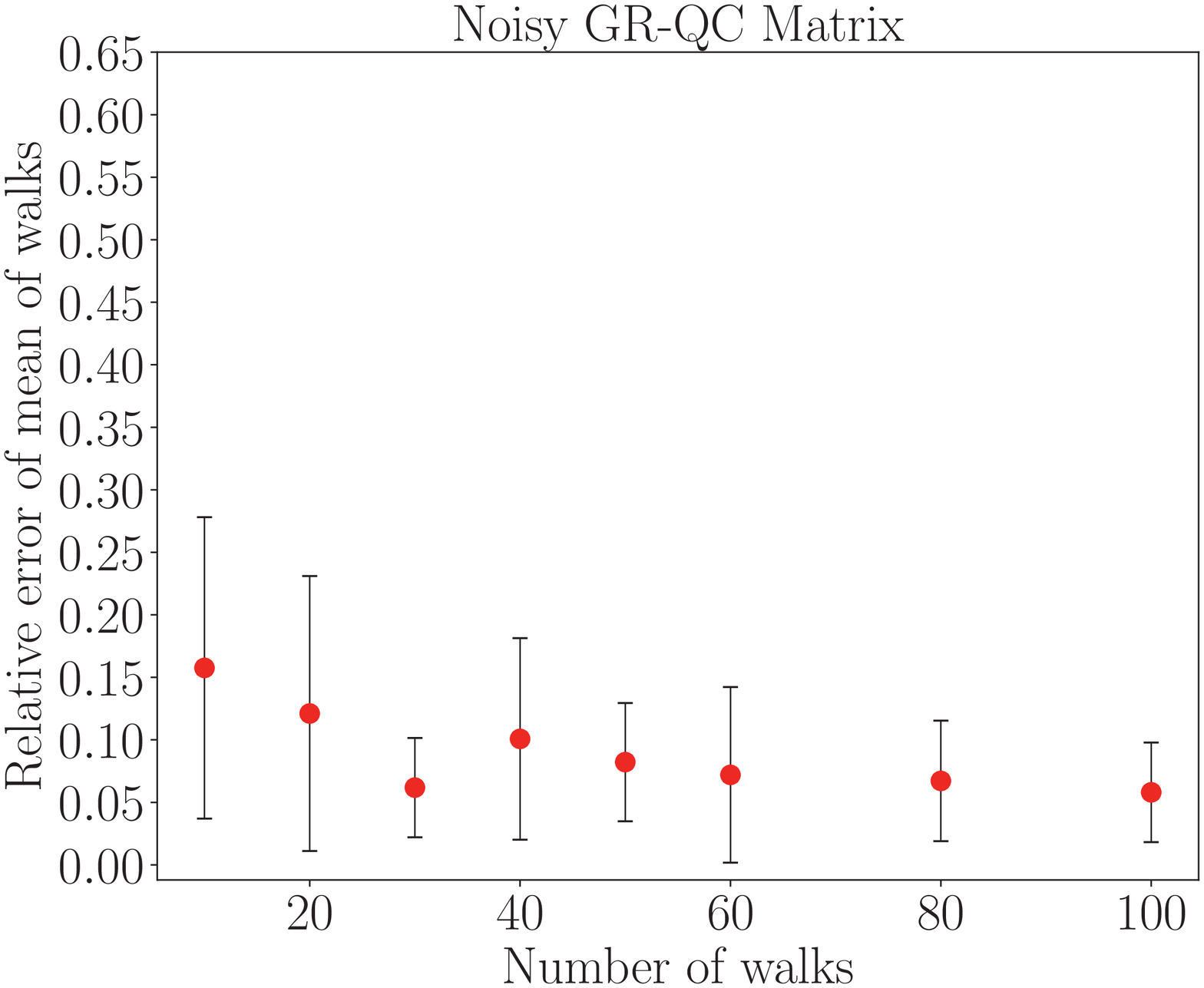}
  \end{minipage}
\caption{Relative error of Algorithm \ref{Alg:row-order} for Schatten $6$-norm of arXiv General Relativity and Quantum Cosmology Collaboration Network: Vary number of walks and plot relative error of the mean of the walks.}
 \label{fig:exp1}
\end{figure}

Recall that the number of independent walks (estimators) is ultimately the space required by Algorithm \ref{Alg:row-order} (upto the space needed to store a row), as they are run the in parallel. Therefore, the left graph shows that the space actually needed to approximate the Schatten $6$-norm of the selected input matrix is significantly smaller than the theoretical bound of Theorem \ref{thm:row-orderAlgSpecial}, which is $O_p{\varepsilon ^{-2}k^(p/2)}\approx 135000$. The other graph shows that the algorithm is robust to small random noise, i.e. it works also for nearly-sparse, where every row and column are dominated by a small amount of entries.

In the second experiment, we use all five collaboration networks -- General Relativity and Quantum Cosmology ($n=5242$), High Energy Physics - Phenomenology ($n=9877$), High Energy Physics - Theory ($n=12008$), Astro Physics ($n=18772$) and Condensed Matter ($n=23133$). For each matrix we compute walks (estimator from Section \ref{sec:Row-orderImportanceSampling}) until the mean of the walks is within $10\%$ of the true Schatten $6$-norm of the matrix. We repeat this process $10$ times for each matrix and plot the median, the first and third quartile of the number of walks for the $10$ trials in Figure \ref{fig:exp2}.Since in the second and third experiments, most of the outputs of the 10 trials are concentrated around the median except for very few trials (one or two) which are very large outliers. Hence, we chose to output the first and third quartiles indicating the output of the majority of the trials.

\begin{figure}[H]
  \centering
  \includegraphics[width= 0.49\linewidth]{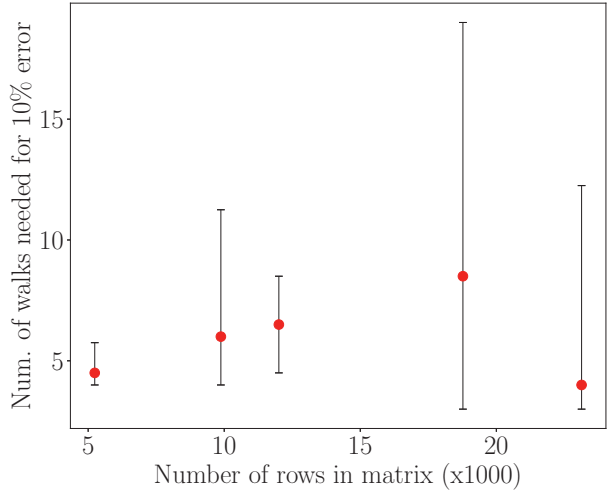}
\caption{Number of walks to $(1\pm 0.1)$-approximate Schatten $6$-norm of $5$ different matrices from arXiv Physics Collaboration Network.}
\label{fig:exp2}
\end{figure}

The above figure shows that indeed calculating the space needed to approximate the Schatten norms using our algorithm is independent of the matrix dimension. Again, as in Figure \ref{fig:exp1}, it is easy to see that the number of estimators needed to approximate the Schatten $6$-norm of the chosen matrices is several orders of scale better than the theoretical bound.

In our third experiment we compute the number of walks needed for the mean of the walks to be within $10\%$ of the true Schatten $p$-Norm of the GR-QC matrix for different values of $p$. We vary the value of $p$ and, for each value of $p$, compute the number of walks needed for $10$ trials and plot the median, first and third quartile of the $10$ trials in Figure \ref{fig:exp3}.

\begin{figure}[H]
  \centering
  \includegraphics[width= 0.49\linewidth]{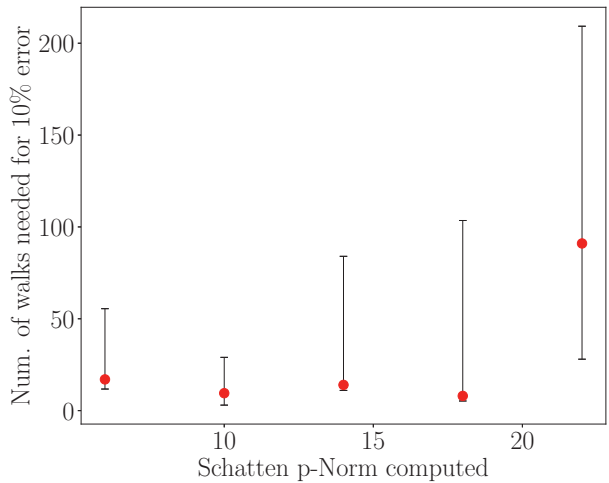}
\caption{Number of walks to $(1\pm 0.1)$-approximate Schatten $p$-norm for arXiv General Relativity and Quantum Cosmology Collaboration Network (GR-QC) for different values of $p \in 2\Z^+$.}
\label{fig:exp3}
\end{figure}

The last figure follows the previous figures, and shows that again the numerical results are much better than the theoretical bounds, in this case in the dependence on $p$. Although there is an expected increase in space as $p$ grows, it is not rapid, and in particular is not exponential. This means, for example, that the space needed to approximate other spectral functions, as explained in Section \ref{sec:Applications}, would be small, suggesting that our algorithm would be practical for such tasks.

\bibliographystyle{alphaurlinit}
\bibliography{Refs}

\appendix
\section{Appendix}

\subsection{Proof of Fact \ref{fact:lengthLeqSchatten-p}}\label{app:lengthLeqSchatten-p}
Let $M = AA^\top$ be a PSD matrix, with eigenvalues $\lambda_1\geq\ldots\geq\lambda_n\geq 0$. Let $\vec{m}, \vec{\lambda} \in \R^n$ be the vectors corresponding to the diagonal entries of $M$ and the eigenvalues of $M$ respectively, both in non-increasing order. Then, by Schur-Horn theorem (Theorem 4.3.26 in \cite{HJ85}), $\vec{\lambda}$ weakly majorizes $\vec{m}$, i.e. $\sum_{i = 1}^r \lambda_i \geq \sum_{i = 1}^r m_i$ for all $r \in [n]$.

Since $f(y) = \sum_{i = 1}^n y_i^t$ is a Schur-convex function for $y \in \R^n$ and $t \geq 1$, we have that $\sum_{i = 1}^n \lambda_i^t \geq \sum_{i = 1}^n m_i^t$. The statement follows from the fact that $\sum_{i = 1}^n \lambda_i^t =\spn{AA^\top}{t}= \spn{A}{2t}$ and $m_i = \|a_i\|_2^{2}$ for all $i \in [n]$.

\subsection{Proof of Theorem \ref{thm:ImportanceSampling}}\label{App:ImpSamplingProof}
It is easy to see that the estimator is unbiased; $\Exp{}{\hat{Z}} = \sum_{i  \in [n]} \frac{z_i}{\tau_i} \cdot \tau_i = z$. Bounding the variance can be done as follows,
\begin{align*}
\Var({\hat{Z}}) &\leq \Exp{}{(\hat{Z})^2} = \sum_{i \in [n]} \left(\frac{z_i}{\tau_i} \right)^2 \tau_i = \sum_{i \in [n]} \left(\frac{|z_i|}{\tau_i} \right)^2 \tau_i .
\end{align*}
Since for each $i \in [n]$ we have $\tau_i \geq \frac{|z_i|}{\lambda z}$, we can bound $\Var({\hat{Z}}) \leq \sum_{i \in [n]} (\lambda z)^2 \tau_i  = (\lambda z)^2 $

\subsection{Proof of Lemma \ref{thm:setSampling}}\label{App:TwoSetProof}

The expectation is straight forward. First assume $r=1$:
$$\Exp{}{Y}=\Exp{}{\frac{z_{i,j}}{\tau_i\tau_j}}=\sum_{l\in [n],m\in [n]}\frac{z_{l,m}}{\tau_l \tau_m}\tau_l \tau_m=z$$
and then using the linearity of expectation,
$$\Exp{}{Y}=\frac{1}{r^2}\sum_{u\in [r],v\in [r]}\Exp{}{\frac{z_{i_u,j_v}}{\tau_{i_u}\tau_{j_v}}}=\frac{1}{r^2}\sum_{u\in [r],v\in [r]}z=z$$

For the variance,
\begin{align*}
\mathbb{E} Y^2 &= \frac{1}{r^4}\sum_{u,v}
\left(
\Exp{}{\sum_{\substack{u'\neq u \\ v'\neq v}}\frac{z_{i_u,j_v}}{\tau_{i_u}\tau_{j_v}}\cdot \frac{z_{i_{u'},j_{v'}}}{\tau_{i_{u'}}\tau_{j_{v'}}}}
+\Exp{}{\left(\frac{z_{i_u,j_v}}{\tau_{i_u}\tau_{j_v}}\right)^2}
+\Exp{}{\sum_{u\neq u'}\frac{z_{i_u,j_v}}{\tau_{i_u}\tau_{j_v}}\cdot \frac{z_{i_{u'},j_v}}{\tau_{i_{u'}}\tau_{j_v}}}
+\Exp{}{\sum_{v\neq v'}\frac{z_{i_u,j_v}}{\tau_{i_u}\tau_{j_v}}\cdot \frac{z_{i_u,j_{v'}}}{\tau_{i_u}\tau_{j_{v'}}}}
\right)\\
\intertext{As $i_u$ and $i_{u'}$ are independent for $u\neq u'$, and similarly for $j_v$ and $j_{v'}$ for $v\neq v'$, we get}
&= \frac{1}{r^4}
\left(
r^2(r-1)^2 z^2
+r^2 \sum_{l,m}\frac{z_{l,m}}{\tau_l}\cdot\frac{z_{l,m}}{\tau_m}
+r^2(r-1)\sum_{l,m,m'}\frac{z_{l,m'}}{\tau_l}\cdot{z_{l,m}}
+r^2(r-1)\sum_{l,l',m}\frac{z_{l',m}}{\tau_m}\cdot{z_{l,m}}
\right)
\\
&\leq z^2
+\frac{1}{r^2} \sum_{l,m}\frac{\abs{z_{l,m}}}{\tau_l\tau_m}\cdot{\abs{z_{l,m}}}
+\frac{1}{r}\sum_{l,m,m'}\frac{\abs{z_{l,m'}}}{\tau_l}\cdot{\abs{z_{l,m}}}
+\frac{1}{r}\sum_{l,l',m}\frac{\abs{z_{l',m}}}{\tau_m}\cdot{\abs{z_{l,m}}}
\end{align*}
As the first term is just $(\Exp{}{Y})^2$, it holds that
\begin{align*}
\text{Var}(Y)&\leq \frac{1}{r^2} \sum_{l,m\in N(l)}\frac{\abs{z_{l,m}}}{\tau_l\tau_m}\cdot\abs{z_{l,m}}
+\frac{1}{r}\sum_{l,m,m'\in N(l)}\frac{\abs{z_{l,m'}}}{\tau_l}\cdot{\abs{z_{l,m}}}
+\frac{1}{r}\sum_{m,l\in N(m),l'\in N(m)}\frac{\abs{z_{l',m}}}{\tau_m}\cdot{\abs{z_{l,m}}} \\
\intertext{Recalling that $z_{l,m}=0$ for all $(l,m)\notin E$, we can rewrite the above as}
&= \frac{1}{r^2} \sum_{l,m\in N(l)}\frac{\abs{z_{l,m}}}{\tau_l\tau_m}\cdot\abs{z_{l,m}}
+\frac{1}{r}\sum_{l, m\in N(l),m'\in N(l)}\frac{\abs{z_{l,m'}}}{\tau_l}\cdot{\abs{z_{l,m}}}
+\frac{1}{r}\sum_{m, l\in N(m), l'\in N(m)}\frac{\abs{z_{l',m}}}{\tau_m}\cdot{\abs{z_{l,m}}} \\
\intertext{and using the bound on the probability,}
&\leq \frac{\lambda^2 z}{r^2} \sum_{l,m\in N(l)}\abs{z_{l,m}}
+\frac{\lambda z}{r}\sum_{l,m\in N(l),m'\in N(l)}\abs{z_{l,m}}
+\frac{\lambda z}{r}\sum_{m,l\in N(m),l'\in N(m)}\abs{z_{l,m}}\\
\intertext{Finally, using the bounds on maximum degrees, we get}
&\leq \left(\frac{\lambda^2}{r^2}+\frac{2\lambda\Delta}{r}\right)z \sum_{i,j\in [n]} \abs{z_{i,j}}
\end{align*}

\subsection{Bounding the tail of the Estrada index Taylor expansion (Theorem \ref{Thm:spectralSums})}\label{App:LaplacianEstradaIndexProof}
We bound the tail of the Estrada index Taylor expansion \eqref{eq:EstradaTaylor}, i.e. 
$\left|\sum_{p=m+1}^{\infty}\frac{\tr(A^p)}{p!}\right| \leq \varepsilon \left| \tr(\exp(A))\right|$ for $m=\ceil{(e\theta+1)\log(1/\varepsilon)-1}.$
\begin{align*}
\left|\sum_{p=m+1}^{\infty}\frac{\tr(A^p)}{p!}\right| &\leq\left|\sum_{p=m+1}^{\infty}\frac{\tr(A^{m+1}A^{p-(m+1)})}{(m+1)!(p-(m+1))!}\right|.\\
\intertext{Using $\tr(AB)\leq \|A\|_{S_\infty} \cdot \tr(B)$ which follows from Von Neuman's trace inequality (see \cite{BDKKZ17}),}
&\leq \frac{\|A^{m+1}\|_{S_\infty}}{(m+1)!} \left|\sum_{p=m+1}^{\infty}\frac{\tr(A^{p-(m+1)})}{(p-(m+1))!}\right|, \\
\intertext{and by the bound on the largest eigenvalue and Stirling's formula,}
&\leq \frac{(e\theta)^{m+1}}{(m+1)^{m+3/2}\sqrt{2\pi}} \left|\sum_{p=0}^{\infty}\frac{\tr(A^p)}{p!}\right| \\
&\leq \left(\frac{e\theta}{m+1}\right)^{m+1} \left|\tr(\exp(A))\right|
\end{align*}

Setting $m=\ceil{(e\theta+1)\log(1/\varepsilon)-1}$ and using $(1-x^{-1})^x\leq e^{-1}$ (for $x>0$) guarantees that
$$\left(\frac{e\theta}{m+1}\right)^{m+1} \leq
\left(\frac{e\theta}{(e\theta+1)\log(1/\varepsilon)}\right)^{m+1}
\leq \left(1-\frac{e\theta}{e\theta+1}\right)^{(e\theta+1)\log(1/\varepsilon)}
=\varepsilon.$$

\end{document}